\documentclass{article}

\usepackage{PRIMEarxiv}

\usepackage[utf8]{inputenc} 
\usepackage[T1]{fontenc}    
\usepackage[bookmarkstype=toc, colorlinks=true, linkcolor=black, pdfborder={0 0 0}, bookmarks=true, bookmarksnumbered=true, citecolor=black]{hyperref}    
\usepackage{url}            
\usepackage{booktabs}       
\usepackage{amsfonts}       
\usepackage{nicefrac}       
\usepackage{microtype}      
\usepackage{lipsum}
\usepackage{fancyhdr}       
\usepackage{graphicx}       

\usepackage{xcolor}
\usepackage{amsmath,amssymb,amsfonts}
\usepackage{algorithmic}
\usepackage{graphicx}
\usepackage{textcomp}
\usepackage{subfigure}
\usepackage{mathtools}
\usepackage{amsthm}

\mathtoolsset{showonlyrefs=true}

\newtheorem{theorem}{Theorem}
\newtheorem{lemma}{Lemma}
\newtheorem{assumption}{Assumption}

\newtheorem{remark}{Remark}

\title{Passivity-based sliding mode control for mechanical port-Hamiltonian systems
}

\author{
  Naoki Sakata, Kenji Fujimoto, Ichiro Maruta \\
  Graduate School of Engineering \\
  Kyoto University \\
  \texttt{sakata.naoki.25x@st.kyoto-u.ac.jp} \\
  \texttt{k.fujimoto@ieee.org} \\
  \texttt{maruta@kuaero.kyoto-u.ac.jp} \\
}

\begin{document}
\maketitle

\begin{abstract}
In this work, we propose a new passivity-based sliding mode control method for mechanical port-Hamiltonian systems. Passivity-based sliding mode control (PBSMC) is unification of sliding mode control and passivity-based control. It achieves sliding mode control and Lyapunov stability simultaneously by employing an energy based Lyapunov function. The proposed method gives a family of stabilizing controllers which smoothly interpolates passivity-based control and sliding mode control with free parameters. The freedom is useful to adjust the trade-off between robustness against external disturbances and undesired chattering vibration. In addition, this paper relaxes the restrictive condition which is required in the authors' former result. As a result, we can apply the proposed PBSMC method to trajectory tracking control problems. Furthermore, the robustness of the proposed controller against matched and unmatched disturbances is investigated. Numerical examples demonstrate the effectiveness of the proposed method.
\end{abstract}


\section{Introduction}

\label{sec:introduction}
Port-Hamiltonian systems are a class of nonlinear systems described by Hamilton's canonical equation, and many physical systems are described in this form, e.g., mechanical systems \cite{Shaft2000}, electro-mechanical systems\cite{rodriguez2003}, nonholonomic systems \cite{Fujimoto2012nonholonomic}, and so on. In order to stabilize such systems, passivity-based control techniques are often employed. For the class of systems, the Hamiltonian function, which represents the total energy of the system, is reshaped by an appropriate input so that the resulting closed-loop system has a desired Hamiltonian function which serves as a Lyapunov function candidate. This method is called energy shaping, and many related techniques have been proposed so far, e.g.\cite{Fujimoto2001,Ortega2002}. Kinetic-potential energy shaping (KPES) method \cite{Ferguson2019} is one of these methods for mechanical port-Hamiltonian systems. It allows us to select a class of potential functions whose arguments are both configuration and momentum. Related works have been studied, e.g., \cite{Venkatraman2010,Romero2015}.

On the other hand, sliding mode control is a nonlinear control method belonging to variable structure control, and the controller is known to be robust against model uncertainties and external disturbances. See, e.g., \cite{Slotine1991,Ferrara2019} for details. In this method, the state of the plant system is enforced to be constrained in a subspace called a sliding surface where the state of the system evolves according to the desired dynamics. To achieve sliding mode control, discontinuous high gain feedback is employed to enforce the state to the sliding surface. However, such input induces chattering problems, which may damage the plant system. For this problem, some methods to alleviate the chattering phenomena have been studied, e.g., replacing the discontinuous input with a continuous one \cite{Slotine1991}, high-order sliding mode control \cite{Levant2003}, and super-twisting algorithm \cite{moreno2011}.

Recently, the authors have proposed a prototype of a passivity-based sliding mode controller for mechanical and electro-mechanical port-Hamiltonian systems \cite{Fujimoto2021,Fujimoto2022}. By selecting a non-smooth function as an artificial potential function with kinetic-potential energy shaping, the resulting controller works as a sliding mode controller. Since sliding mode control is achieved in the framework of passivity-based control, Lyapunov stability is ensured even if the input is replaced by a smooth approximation of the sliding mode control law to alleviate chattering phenomena.

In this paper, we proposed a new passivity-based controller that achieves sliding mode control and Lyapunov stability simultaneously. Firstly, we generalize the KPES method so that the closed-loop system has an artificial potential function whose argument is nonlinear with respect to both configuration and momentum, whereas such argument is linear with respect to momentum in the previous methods \cite{Ferguson2019,Fujimoto2021}. Secondly, we show that sliding mode control is realized with Lyapunov stability by selecting an appropriate potential function. In the authors' former result \cite{Fujimoto2021}, the closed-loop system consists of multiple scalar sliding mode subsystems. Such a realization requires a restrictive condition to be satisfied. On the other hand, we prove the stability of multiple sliding modes which can be realized under a relaxed condition. This result makes it possible to apply the proposed method to several control problems, particularly, to trajectory tracking control. In addition, we relax the design condition to achieve sliding mode that allows additional freedom in selecting artificial potential functions. Utilizing such additional freedom, a parameterization of stabilizing controllers consisting of both standard passivity-based (Lyapunov based) controllers and sliding mode ones is obtained. Thus, we can freely adjust the trade-off between robustness against external disturbances and undesired chattering vibration without losing Lyapunov stability. Furthermore, for the control system to which the proposed method is adopted, the robustness against matched and unmatched disturbances is analyzed. A numerical example shows how the proposed controller works.


The remainder of this paper is organized as follows. First, we introduce passivity-based control, kinetic-potential energy shaping method, and sliding mode control in Section~2.  A generalized version of KPES and a new PBSMC method is proposed in Section~3. Section~4 shows application of the proposed method to trajectory tracking control problems. Moreover, the robustness of the proposed controller against matched and unmatched disturbances is investigated in Section~5. The effectiveness of the proposed controller is shown through numerical examples in Section~6. Finally, we conclude this paper in Section~7.

\subsection*{Notation}
The symbol $I_n$ denotes the $n\times n$ identity matrix, and $0_{n}$ is the $n\times n$ matrix of zeros.
For a vector $a$ and a symmetric matrix $B$, $\|a\|^2_B = a^\top B a$.
The symbol $\nabla_{(\cdot)}$ denotes the gradient with respect to $(\cdot)$, that is $ \nabla_x f \equiv \frac{\partial f}{\partial x}^\top = \left( \frac{\partial f}{\partial x_1}, \frac{\partial f}{\partial x_2}, \dots, \frac{\partial f}{\partial x_n} \right)^\top$ with $(x_1, x_2, \dots, x_n )^\top.$
For a vector $x$, the symbol $\|x\|_p$ denotes a $p$-norm of $x$ defined by $\|x\|_p = (|x_1|^p + \dots, + |x_n|^p)^{\frac{1}{p}}$ and, especially, we omit the subscript in the case $p=2$, that is, $\|x\| = \|x\|_2$.


\section{Preliminaries}
In this section, we introduce the backgrounds of port-Hamiltonian systems, kinetic potential energy shaping, and sliding mode control.

\subsection{Port-Hamiltonian systems}
Let us consider a fully-actuated mechanical system described as the following port-Hamiltonian form \cite{Shaft2000},
\begin{align}
    &\begin{pmatrix} \dot{q} \\ \dot{p} \end{pmatrix} = 
    \begin{pmatrix} 0_m & I_m \\ -I_m &  -D_0(q,p) \end{pmatrix}
    \begin{pmatrix} \nabla_q H_0(q,p) \\ \nabla_p H_0(q,p) \end{pmatrix} + 
    \begin{pmatrix} 0_m \\ G_0(q) \end{pmatrix} u, \notag \\
    &H_0(q,p) = \frac{1}{2}p^\top M(q)^{-1} p. \label{Original_PH_sys}
\end{align}
Here $(q^\top, p^\top)^\top \in\mathbb{R}^{2m}$ denotes the state vector of the plant system, $q\in \mathbb{R}^m$ and $p\in \mathbb{R}^m$ denote the configuration vector and the momentum vector respectively, and $u\in \mathbb{R}^m$ denotes the input vector.
The symbol $M(q) = M(q)^\top\in\mathbb{R}^{m\times m}$ denotes the inertia matrix and it is a positive definite matrix.
The matrix $D_0(q,p)\in \mathbb{R}^{m\times m}$ is the damping matrix which is positive semi-definite.
The matrix $G_{0}(q)\in \mathbb{R}^{m\times m}$ denotes the full rank input mapping matrix.
The symbol $H_0(q,p)\in\mathbb{R}$ is called the Hamiltonian function which represents the total energy of the system.

\subsection{Momentum transformation and kinetic-potential energy shaping}
For a class of mechanical port-Hamiltonian systems, the change of coordinates in momenta is often used so that the kinetic energy becomes independent of $q$, e.g.\cite{Venkatraman2010,Romero2015,Ferguson2019}.
Kinetic-potential energy shaping method \cite{Ferguson2019} is one of these techniques, and we can select a potential function whose arguments are both configuration and momentum.

Let us consider the following coordinate transformation to the plant system \eqref{Original_PH_sys}
\begin{align}
    x = \begin{pmatrix} q \\ \eta \end{pmatrix} \equiv \begin{pmatrix} q \\ T(q)^\top p \end{pmatrix}, \label{Def_suqare_coordinate}
\end{align}
where $T(q)\in\mathbb{R}^{m\times m}$ is a nonsingular matrix satisfying
\begin{align}
    T(q)T(q)^\top = M(q)^{-1}. \label{Definition_T}
\end{align}
Then the system \eqref{Original_PH_sys} is transformed into the following one with a new Hamiltonian function $H(\eta)=(1/2)\|\eta\|^2$
\begin{align}
    \begin{pmatrix} \dot{q} \\ \dot{\eta} \end{pmatrix} =
    \underbrace{%
        \left( \begin{array}{@{\,}c@{\,}c@{\,}} 0_n & T(q) \\ -T(q)^\top &  -D(q,\eta) \end{array} \right)\!
    }_{J(x)}
    \begin{pmatrix} \nabla_{q} H(\eta) \\ \nabla_{\eta} H(\eta) \end{pmatrix} +
    \left( \begin{array}{@{\,}c@{}}
        0_n \\ G(q)
    \end{array} \right) u, \label{Square_PH_sys}
\end{align}
where $G(q) = T(q)^\top G_0(q)$ and $D(q, \eta)$ is a matrix satisfying $D(q,\eta) + D(q,\eta)^\top \succeq 0$.
Note that $D(q, \eta)$ consists of a gyroscopic term and a damping term.
Applying such a coordinate transformation eliminates the $M(q)^{-1}$ term in the Hamiltonian function $H_0(q,p)$, and then the new Hamiltonian function $H(\eta)$ becomes independent of $q$.
Thus, an appropriate modification of the upper left block of the structure matrix $J(x)$ makes it possible to choose a potential function that depends on both configuration $q$ and momentum $\eta$ without changing the kinematics.
\subsection{Sliding mode control}
Sliding mode control is a nonlinear control method and it belongs to a variable structure control.
See, e.g., \cite{Slotine1991,Ferrara2019} for details.
In the sliding mode control, the state of the plant system goes from the initial point towards a subspace of the state space called a sliding surface, which is called reaching mode.
After reaching the sliding surface, the state variable evolves along the desired dynamics on the surface, which is called sliding mode.
The control input is designed so that the state reaches the sliding surface in a finite time and stays there.
It often employs a discontinuous high gain input so that the state variable is constrained to the sliding surface.

Here, we consider a general input-affine nonlinear system 
\begin{align}
    \dot{x} = f(x) + g(x)u \label{BSMC_Eq}
\end{align}
with $x\in\mathbb{R}^n$ and $u\in\mathbb{R}^m$.
The symbol $\sigma(x)\in \mathbb{R}^m$ denotes a sliding variable.
The sliding surface is given by $\sigma(x)=0$, which is designed to achieve the desired dynamics in the sliding mode.
In many cases, sliding mode controllers consist of discontinuous functions so that the closed-loop systems include the following dynamics
\begin{align}
    \dot{\sigma}_i = -\lambda_i \,\mathrm{sgn}\, \sigma_i, \quad i = 1,2,\dots, m, \label{BSMC_Dynamics_sigma}
\end{align}
where $\lambda_i$'s are positive constants.
Here, $\mathrm{sgn}(\cdot)$ is the signum function defined by
\begin{align*}
    \mathrm{sgn}\, z \begin{cases} =1 & (z>0) \\ \in[-1,1] & (z=0) \\ =-1 & (z<0) \end{cases}.
\end{align*}
If its argument is a vector $x = (x_1, \dots , x_n)^\top \in\mathbb{R}^n$, then
\begin{align*}
    \mathrm{sgn}\,x = (\mathrm{sgn}\,x_1, \dots, \mathrm{sgn}\,x_n)^\top.
\end{align*}
By using such input, we can ensure that the sliding variable $\sigma(x)$ converges to zero in a finite time.
Then the state variable evolves along the desired dynamics on the sliding surface $\sigma(x) = 0$.

\section{Stabilization of passivity-based sliding mode control}
This section gives the main result of the paper.
A novel sliding mode controller for port-Hamiltonian systems is proposed that ensures Lyapunov stability.

\subsection{Generalized kinetic potential energy shaping}
In this subsection, we show one of the main results of this paper that the generalized version of KPES is proposed.
The resulting closed-loop port-Hamiltonian systems have special potential functions depending nonlinearly on both configuration $q$ and momentum $\eta$.
The following lemma gives such closed-loop systems.

\begin{lemma}\label{lem_genKPES}
Consider the system \eqref{Square_PH_sys} with any function $U:\mathbb{R}^m\to \mathbb{R}$, any smooth vector function $\phi:\mathbb{R}^m \times \mathbb{R}^m \to \mathbb{R}^m$ satisfying $\phi(0,0)=0$, and any matrix $D_\mathrm{d}(q,\eta)\in \mathbb{R}^{m\times m}$.
Suppose that the partial derivatives of $\phi(q,\eta)$ with respect to $q$ and $\eta$ are nonsingular, that is, there exists
\begin{align*}
    \frac{\partial \phi(q,\eta)}{\partial q}^{-1}, \frac{\partial \phi(q,\eta)}{\partial \eta}^{-1}
\end{align*}
for all $q$ and $\eta$.
Then, the feedback
\begin{align}
        u &= -G(q)^{-1}\Bigg\{(D_\mathrm{d}(q,\eta) - D(q,\eta))\eta + \left( T(q)^\top \frac{\partial \phi(q,\eta)}{\partial q}^\top  +  D_\mathrm{d}(q,\eta) \frac{\partial \phi(q,\eta)}{\partial \eta}^\top\right) \nabla_{\sigma} U(\sigma) \Biggr\} \label{General_KPES_input}
\end{align}
converts \eqref{Square_PH_sys} into the closed-loop Hamiltonian system 
\begin{align}
    &\begin{pmatrix} \dot{q} \\ \dot{\eta} \end{pmatrix} = 
    \underbrace{\left(\begin{array}{@{\,}c@{\,}c@{\,}}
        -T(q) \frac{\partial \sigma(q,\eta)}{\partial \eta}^{\top}\frac{\partial \sigma(q,\eta)}{\partial q}^{-\top} & T(q) \\ -T(q)^\top & -D_{\mathrm{d}}(q,\eta)
    \end{array}\right)}_{J_{\mathrm{d}}(x)}
    \begin{pmatrix} \nabla_{q} H_\mathrm{d} \\ \nabla_{\eta} H_\mathrm{d}\end{pmatrix},  \notag \\
    & H_\mathrm{d} (q,\eta) = \frac{1}{2}\|\eta\|^2 + U(\phi(q,\eta)), \label{General_KPES_sys}
\end{align}
where $\sigma \equiv \phi(q,\eta)$.
Furthermore, if the function $U$ is positive definite and smooth, and if 
\begin{align}
    &\frac{\partial \phi(q,\eta)}{\partial q}T(q)\frac{\partial \phi(q,\eta)}{\partial \eta}^\top + \frac{\partial \phi(q,\eta)}{\partial \eta}T(q)^\top \frac{\partial \phi(q,\eta)}{\partial q}^\top \succ 0 \label{condition_thm1_posLam},\\
    &D_\mathrm{d}(q,\eta) + D_\mathrm{d}(q,\eta)^\top \succ 0 \label{condition_thm1_posDd}
\end{align}
hold, then the origin of the transformed system \eqref{General_KPES_sys} is asymptotically stable with the Lyapunov function $H_\mathrm{d}(q,\eta)$.
\end{lemma}

\begin{proof}
Firstly, let us prove the former part of Lemma~\ref{lem_genKPES}.
The partial derivatives of $U(q,\eta)$ with respect to $q$ and $\eta$ are calculated as
\begin{align}
    \frac{\partial U(\phi(q,\eta))}{\partial q} = \frac{\partial U(\sigma)}{\partial \sigma}
    \frac{\partial \phi(q,\eta)}{\partial q}, \label{rnd_U_rnd_q}\\
    \frac{\partial U(\phi(q,\eta))}{\partial \eta} = \frac{\partial U(\sigma)}{\partial \sigma}
    \frac{\partial \phi(q,\eta)}{\partial \eta} \label{rnd_U_rnd_p}.
\end{align}
Substituting \eqref{General_KPES_input} into \eqref{Square_PH_sys} and using \eqref{rnd_U_rnd_q} and \eqref{rnd_U_rnd_p}, we obtain
\begin{align}
    \dot{q} &= T(q)\eta \notag \\
    &= -T(q)\frac{\partial \phi(q,\eta)}{\partial \eta}^\top\frac{\partial \phi(q,\eta)}{\partial q}^{-\top} \nabla_q H_{\mathrm{d}} + T(q)\nabla_p H_{\mathrm{d}} \label{dot_Q_lemm1}\\ 
    \dot{p} &= -D(q,\eta)\eta + G(q)u \notag \\
    &= -T(q)^\top \frac{\partial \phi(q,\eta)}{\partial q}^\top \nabla_{\sigma}U(\sigma) -D_{\mathrm{d}}(q,\eta)\left(\frac{\partial \phi(q,\eta)}{\partial \eta}^\top\nabla_{\sigma}U(\sigma)  + \eta\right) \notag \\ 
    &= -T(q)^\top \nabla_q H_{\mathrm{d}}(q,\eta) -D_{\mathrm{d}}(q,\eta)\nabla_\eta H_{\mathrm{d}}(q,\eta) \label{dot_P_lemm1}.
\end{align}
Equations \eqref{dot_Q_lemm1} and \eqref{dot_P_lemm1} agree with the dynamics of \eqref{General_KPES_sys}.
Thus, the input \eqref{General_KPES_input} converts the system \eqref{Square_PH_sys} into the closed-loop system \eqref{General_KPES_sys}.

Next, we consider the latter part of the lemma.
Let us consider the $H_{\mathrm{d}}(x)$ is a Lyapunov function candidate.
The structure matrix $J_{\mathrm{d}}(x)$ satisfies
\begin{align}
    J_{\mathrm{d}}(x) + J_{\mathrm{d}}(x)^\top  
    &= -\mathrm{diag} \Biggl( 
    \frac{\partial \phi}{\partial q}^{-1}\Biggl( 
    \frac{\partial \phi}{\partial q}T\frac{\partial \phi}{\partial \eta}^\top
    + 
    \frac{\partial \phi}{\partial \eta}T^\top \frac{\partial \phi}{\partial q}^\top
    \Biggr)
    \frac{\partial \phi}{\partial q}^{-\top}, D_\mathrm{d} + D_\mathrm{d}\Biggr) \prec 0,
    \label{negative_Jd} 
\end{align}
due to the assumptions \eqref{condition_thm1_posLam} and \eqref{condition_thm1_posDd}.
Then it follows from the positive definiteness of $U$ and \eqref{negative_Jd} that
\begin{align*}
    H_{\mathrm{d}} &\succ 0, \\
    \frac{\mathrm{d}H_{\mathrm{d}}}{\mathrm{d}t} 
    &= \frac{\partial H_{\mathrm{d}}}{\partial x}\frac{\mathrm{d}x}{\mathrm{d}t}
     = \nabla_x H_{\mathrm{d}}(x)^\top J_{\mathrm{d}}(x)\nabla_x H_{\mathrm{d}}(x)\\
    &= \frac{1}{2}\nabla_x H_{\mathrm{d}}(x)^\top ( J_{\mathrm{d}}(x) + J_{\mathrm{d}}(x)^\top )\nabla_x H_{\mathrm{d}}(x)\\
    &\prec 0.
\end{align*}
Therefore, the origin of the closed-loop system \eqref{General_KPES_sys} is asymptotically stable with the Lyapunov function $H_{\mathrm{d}}$.
This completes the proof. 
\end{proof}

This lemma shows that the closed-loop port-Hamiltonian system with an artificial potential function $U(\phi(q,\eta))$ is obtained by the feedback input \eqref{General_KPES_input}.
Since the functions $U$ and $\phi$ are free parameters, we can construct various stabilizing controllers with Lyapunov stability by adjusting them.
Such freedom will be used to realize sliding mode in the next subsection.

\subsection{Passivity-based sliding mode control}
In the previous subsection, we obtain the closed-loop port-Hamiltonian systems with free parameters $\phi$ and $U$ for selecting a Lyapunov function candidate.
This subsection shows that by selecting these parameters appropriately the resulting controller works as a sliding mode controller with Lyapunov stability.
Here, let us define the sliding variable by
\begin{align}
    \sigma = \phi(q,\eta) . \label{Definition_sigma}
\end{align}
To derive the proposed passivity-based sliding mode controller, let us consider the following assumptions.
\begin{assumption}\label{asm_positive_Lambda}
    The symmetric matrix $\Lambda(q,\eta)$ defined by 
    \begin{align}
        \Lambda(q,\eta) &\equiv \frac{\partial \phi(q,\eta)}{\partial q}T(q)\frac{\partial \phi(q,\eta)}{\partial \eta}^\top  + \frac{\partial \phi(q,\eta)}{\partial \eta} T(q)^\top \frac{\partial \phi(q,\eta)}{\partial q}^\top \label{DefLambda}
    \end{align}
    is uniformly positive definite, that is, there exists a constant $\varepsilon > 0$ such that 
    \begin{align}
        \Lambda(q,\eta) \succ \varepsilon I_m \succ 0, \ \forall q, \eta.
    \end{align}
\end{assumption}
\begin{assumption}\label{asm_U}
    Let $U:\mathbb{R}^m \to \mathbb{R}$ be a positive definite function.
    This function satisfies
    \begin{align}
        \|\nabla_\sigma U(\sigma)\| \geq c\, U(\sigma)^\rho
    \end{align}
    except $\sigma=0$, where $c$ is a positive constant and $\rho$ is a constant satisfying $0 \leq \rho < 1/2$.
\end{assumption}

The following theorem gives a new passivity-based sliding mode controller.
\begin{theorem}\label{thm_stabilization}
    Consider the system \eqref{Square_PH_sys} with any function $U: \mathbb{R}^m \to \mathbb{R}$ and a vector function $\phi: \mathbb{R}^m \times \mathbb{R}^m \to \mathbb{R}^m$ satisfying $\phi(0,0)=0$.
    Suppose that the partial derivatives of $\partial \phi(q,\eta)/\partial q$ and $\partial \phi(q,\eta)/\partial \eta$ of $\phi(q,\eta)$ are nonsingular for any $q$ and $\eta$.
    Then the feedback input 
    \begin{align}
        u &= - G(q)^{-1}\frac{\partial \phi(q,\eta)}{\partial \eta}^{-1}\Lambda(q,\eta)\nabla_\sigma U(\sigma) + G(q)^{-1}\left(D(q,\eta)-\frac{\partial \phi(q,\eta)}{\partial \eta}^{-1}\frac{\partial \phi(q,\eta)}{\partial q}T(q)\right) \eta \label{feedback_genPBSMC}
    \end{align}
    converts the system \eqref{Square_PH_sys} into the following closed-loop system
    \begin{align}
            &\begin{pmatrix} \dot{q} \\ \dot{\eta} \end{pmatrix} = 
            \begin{pmatrix} -T\frac{\partial \phi}{\partial \eta}^\top\frac{\partial \phi}{\partial q}^{-\top} & T \\ -T^\top & -\frac{\partial \phi}{\partial \eta}^{-1}\frac{\partial \phi}{\partial q}T \end{pmatrix}
            \begin{pmatrix} \nabla_q H_{\mathrm{smc}} \\ \nabla_\eta H_{\mathrm{smc}} \end{pmatrix}, \notag\\
            &H_{\mathrm{smc}}(q,\eta) = \frac{1}{2}\|\eta\|^2 + U(\phi(q,\eta)). \label{genPBSMC_sys}
    \end{align}
    Furthermore, if Assumptions \ref{asm_positive_Lambda} and \ref{asm_U} hold, then the sliding variable $\sigma \equiv \phi(q,\eta)$ is enforced to converge to zero after a finite time transient, and the origin of the closed-loop system \eqref{genPBSMC_sys} is asymptotically stable with the Lyapunov function $H_{\mathrm{smc}}$.
\end{theorem}
\begin{proof}
    Since the partial derivatives of $U(q,\eta)$ with respect to $q$ and $\eta$ are calculated as \eqref{rnd_U_rnd_q} and \eqref{rnd_U_rnd_p}, it is proved by a direct calculation that the port-Hamiltonian system \eqref{Square_PH_sys} is transformed into the closed-loop system \eqref{genPBSMC_sys} by the input \eqref{feedback_genPBSMC}.
    This proves the first part of this theorem.
    
    Next, to prove finite time convergence of the sliding variable $\sigma$, let us apply the following coordinate transformation
    \begin{align}
        \begin{pmatrix} q \\ \eta \end{pmatrix} \mapsto 
        \begin{pmatrix} \sigma \\ \eta \end{pmatrix} = 
        \begin{pmatrix} \phi(q,\eta) \\ \eta \end{pmatrix} \label{coordinate_trans_to_sigma}.
    \end{align}
    Then the closed-loop system is represented as follows:
    \begin{align*}
        \begin{pmatrix} \dot{\sigma} \\ \dot{\eta} \end{pmatrix} &= 
        \begin{pmatrix} -\Lambda & 0 \\ -
        \frac{\partial \phi}{\partial \eta}^{-1} \Lambda & -\frac{\partial \phi}{\partial \eta}^{-1}\frac{\partial \phi}{\partial q}T \end{pmatrix}
        \begin{pmatrix} \nabla_{\sigma} H_{\mathrm{smc}} \\ \nabla_{\eta} H_{\mathrm{smc}} \end{pmatrix}, \notag \\
        H_{\mathrm{smc}} &= \frac{1}{2}\|\eta\|^2 + U(\sigma).
    \end{align*}
    We can see that the dynamics of sliding variable $\sigma$ is 
    \begin{align}
        \dot{\sigma} = -\Lambda(q,\eta)\nabla_{\sigma} U(\sigma) \label{DotSigma}.
    \end{align}
    Now, let us consider $U(\sigma)$ as a Lyapunov function candidate and prove that there exists a constant $a\geq 1$ satisfying
    \begin{align}
        \dot{U} & \leq - \frac{\varepsilon c^2}{a^2}U^{2\rho} \label{inequarity_dot_U}
    \end{align}
    along the closed loop system \eqref{genPBSMC_sys}.
    Equation \eqref{inequarity_dot_U} is proved in multiple cases such as reaching mode and sliding modes. 
    First of all, for the reaching mode, i.e., $\sigma_i \neq 0$ for $\forall i$, the following inequality holds from Assumptions~\ref{asm_positive_Lambda} and \ref{asm_U}
    \begin{align}
        \dot{U} &= \nabla_\sigma U(\sigma)^\top (-\Lambda(q,\eta)\nabla_\sigma U(\sigma)) \notag \\
                &\leq -\varepsilon\|\nabla_\sigma U(\sigma)\|^2 \notag \\
                &\leq -\varepsilon c^2 U(\sigma)^{2\rho} \label{inequarity_dot_U_rm}.
    \end{align}
    Thus, \eqref{inequarity_dot_U} holds with $a=a_0=1$ .
    Next for the sliding modes, i.e, there exists some $i$'s for which $\sigma_i=0$.
    Suppose there are $k$ sub-sliding modes and $m-k$ sub-reaching modes, i.e., $k$ elements of sliding variable vector $\sigma$ are enforced to be zero and $m-k$ elements of $\sigma$ are not yet zero for $1\leq k\leq m-1$.
    Let us denote $\sigma_{\mathrm{sm}} \equiv (\sigma_1, \dots, \sigma_k)^\top, \sigma_{\mathrm{rm}} \equiv (\sigma_{k+1}, \dots, \sigma_m)^\top$ and suppose $\sigma_{\mathrm{sm}} = \dot{\sigma}_{\mathrm{sm}} = 0$ (sliding mode) and $\sigma_{\mathrm{rm}} \neq 0$ (reaching mode) for simplicity.
    Then the dynamics of sliding variable \eqref{DotSigma} is represented as
    \begin{align}
        \frac{\mathrm{d}}{\mathrm{d}t}
        \begin{pmatrix}
            \sigma_{\mathrm{sm}} \\ \sigma_{\mathrm{rm}} 
        \end{pmatrix} &= -
        \begin{pmatrix}
            \Lambda_{11}(q,\eta) & \Lambda_{12}(q,\eta) \\ 
            \Lambda_{12}(q,\eta)^\top & \Lambda_{22}(q,\eta)
        \end{pmatrix}
        \begin{pmatrix}
            \nabla_{\sigma_{\mathrm{sm}}} U(\sigma) \\ \nabla_{\sigma_{\mathrm{rm}}} U(\sigma)
        \end{pmatrix} \notag \\ &= 
        \begin{pmatrix} 0 \\ \ast \end{pmatrix},\label{equivalent_system}
    \end{align}
    where $\ast$ is an arbitrary value.
    It holds from \eqref{equivalent_system} and Assumption~\ref{asm_U} that
    \begin{align}
        c U(\sigma)^\rho &\leq \|\nabla_\sigma U(\sigma) \| \notag \\ &= 
        \Biggl\| \underbrace{\begin{pmatrix} -\Lambda_{11}(q,\eta)^{-1}\Lambda_{12}(q,\eta) \\ I_{m-k}  \end{pmatrix}}_{L^k(q,\eta)} \nabla_{\sigma_{\mathrm{rm}}} U(\sigma)\Biggr\| \notag \\
        &= \|L^k(q,\eta)\|\|\nabla_{\sigma_{\mathrm{rm}}} U(\sigma)\| \notag \\
        &\leq l^{k}_{\max} \|\nabla_{\sigma_{\mathrm{rm}}} U(\sigma)\|, \label{relation_nabla_sigma_rm_U}
    \end{align}
    where $l^k_{\max} \geq 1$ is an upper bound of $\|L^{k}(q,\eta)\|$ in a neighborhood of the origin.
    The equivalent sliding mode control system is described as
    \begin{align}
        \dot{\sigma}_{\mathrm{rm}} &= 
        -(\underbrace{\Lambda_{22} - \Lambda_{12}^\top \Lambda_{11}^{-1} \Lambda_{12}}_{\Lambda/\Lambda_{11}(q,\eta)})
        \nabla_{\sigma_{\mathrm{rm}}} U(\sigma) \label{equivalent_sigma_rm}.
    \end{align}
    Calculating the time derivative of $U(\sigma)$ with \eqref{relation_nabla_sigma_rm_U} and \eqref{equivalent_sigma_rm}, we obtain
    \begin{align}
        \dot{U} &= \left( \frac{\partial U(\sigma)}{\partial \sigma_{\mathrm{sm}}}, \frac{\partial U(\sigma)}{\partial \sigma_{\mathrm{rm}}} \right)
        \begin{pmatrix} \dot{\sigma}_{\mathrm{sm}} \\ \dot{\sigma}_{\mathrm{rm}}\end{pmatrix} \notag \\
        &= -\nabla_{\sigma_{\mathrm{rm}}}U^\top (\Lambda/\Lambda_{11}) \nabla_{\sigma_{\mathrm{rm}}}U \notag \\
        &\leq -\varepsilon \|\nabla_{\sigma_{\mathrm{rm}}}U(\sigma)\|^2 \leq - \frac{\varepsilon c^2}{(l^k_{\max})^2} U(\sigma)^{2\rho}, \label{inequarity_dot_U_subsm}
    \end{align}
    where we use the fact that $\Lambda/\Lambda_{11}(q,\eta) \succ \varepsilon I_{m-k}$ which is derived from Assumption~\ref{asm_positive_Lambda} and the decomposition of $\Lambda(q,\eta)$ as follows:
    \begin{align*}
        \Lambda = 
        \begin{pmatrix} I_k & 0 \\ \Lambda_{12}^\top \Lambda_{11}^{-1}& I_{m-k} \end{pmatrix}
        \begin{pmatrix} \Lambda_{11} & 0 \\ 0 & \Lambda/\Lambda_{11} \end{pmatrix}
        \begin{pmatrix} I_k & \Lambda_{11}^{-1}\Lambda_{12}\\ 0 & I_{m-k} \end{pmatrix}.
    \end{align*} 
    Thus, the inequality \eqref{inequarity_dot_U} holds with $a=a_k = l_{\max}^k$ even if any sub-sliding mode occurs.
    Hence, it can be proved that \eqref{inequarity_dot_U} holds in both the reaching mode and any sub-sliding mode by redefining the maximum value of $a_k$'s in each case as $a = \max_k a_k$.
    By integrating \eqref{inequarity_dot_U} with respect to time, we obtain
    \begin{align}
        U(\sigma(t)) \leq \left(-\frac{\varepsilon c^2}{a^2} (t-t_0) + U(\sigma(t_0))^{1-2\rho}\right)^{\frac{1}{1-2\rho}}.
    \end{align}
    Then, $U(\sigma)$ becomes zero within a finite time $t_0 + a^2 U(\sigma(t_0))^{1-2\rho}/(\varepsilon c^2)$. 
    Therefore, the sliding variable $\sigma$ converges to zero in a finite time. 

    In the last part of the proof, we will prove asymptotic stability of the closed-loop system.
    The time derivative of the Hamilton function $H_{\mathrm{smc}}$ in the reaching mode $\sigma_i \neq 0$ is calculated as
    \begin{align*}
        \dot{H}_{\mathrm{smc}} = -\frac{1}{4}\left\|\frac{\partial \sigma}{\partial \eta}^{-\top}\eta+2\nabla_\sigma U\right\|^2_{\Lambda} - \frac{1}{4}\left\|\frac{\partial \sigma}{\partial \eta}^{-\top}\eta\right\|_{\Lambda}^2 \prec 0.
    \end{align*}
    In particular, if there are sub-sliding modes where
    \begin{align*}
        \dot{\sigma}_{\mathrm{sm}} = 0, \dot{\sigma}_{\mathrm{rm}} \neq 0, 
    \end{align*}
    holds, then it follows from \eqref{equivalent_system} that 
    \begin{align*}
        \dot{H}_{\mathrm{smc}} 
        &= -\frac{1}{4}\left\|\Lambda\,\frac{\partial \phi}{\partial \eta}^{-\top}\eta + 2\begin{pmatrix} 0 \\ \dot{\sigma}_\mathrm{rm} \end{pmatrix}\right\|^2_{\Lambda^{-1}} - \frac{1}{4}\left\|\frac{\partial \phi}{\partial \eta}^{-\top}\eta\right\|_{\Lambda}^2  \\
        &\prec 0,
    \end{align*}
    in the state space of the equivalent system where $\sigma_{\mathrm{sm}} = 0$.
    Moreover, it holds in the sliding mode $\sigma=0$ that
    \begin{align*}
        \dot{\sigma} = \Lambda(q,\eta)\nabla_\sigma U(\sigma) = 0.
    \end{align*}
    Then the time derivative of $H_{\mathrm{smc}}$ is calculated as
    \begin{align*}
        \dot{H}_{\mathrm{smc}} = - \frac{1}{2}\left\|\frac{\partial \phi}{\partial \eta}^{-\top}\eta\right\|_{\Lambda}^2 \prec 0
    \end{align*}
    on the sliding surface $\sigma = 0$.
    Thus, the time derivative of $H_{\mathrm{smc}}$ is negative in reaching mode, sub-sliding mode, and sliding mode.
    Therefore, the origin of the closed-loop system is asymptotically stable with the Lyapunov function $H_{\mathrm{smc}}$ for all cases.
    This completes the proof.
\end{proof}

This theorem shows that the proposed passivity-based sliding mode controller achieves sliding mode control and Lyapunov stability simultaneously.
It also realizes sliding mode control with arbitrary smoothness (arbitrary convergence speed), since a parameterization of stabilizing controllers consisting of both standard passivity-based controllers and sliding mode ones is obtained.
In addition to that, we can freely adjust the trade-off between robustness against external disturbances and chattering vibration because the input can be changed from discontinuous inputs to continuous ones as depicted in Fig~\ref{Variable_potential_and_input}.
The reason for these advantages is that two Lyapunov functions $H_{\mathrm{smc}}(q,\eta)$ and $U(\sigma)$ are used simultaneously to ensure stability.
The Hamiltonian function $H_\mathrm{smc}(q,\eta)$ guarantees Lyapunov stability of the whole system, while $U(\sigma)$ ensures finite time convergence of the sliding variable $\sigma$.
Examples of the potential function $U(\sigma)$ satisfying Assumption~\eqref{asm_U} are
\begin{align}
    U(\sigma) &= k\|\sigma\|_s^r,\ k>0,\ 1\leq r<2,\ 1\leq s, \label{norm-potential} \\
    U(\sigma) &= \alpha\|\sigma\|_1 + \frac{\beta}{2}\|\sigma\|^2,\ \alpha>0,\ \beta>0.
\end{align}
In particular, the potential function \eqref{norm-potential} reduces chattering phenomenon by selecting a large value of $r$ and improves the behavior of convergence of the sliding variable with an appropriate value of $s$. 
See \cite{Sakata2023} for details.

\begin{remark}\label{rmk_unnecessarity_of_diagonal}
    The prototype of a passivity-based sliding mode controller \cite{Fujimoto2021} has realized the multiple scalar sliding mode subsystems where each component of $\sigma$ achieves \eqref{BSMC_Dynamics_sigma}.
    So, it requires $\Lambda$ to be diagonal because the dynamics of the sliding variables $\sigma_i$'s are decoupled as in \eqref{BSMC_Dynamics_sigma}.
    However, this condition is difficult to be satisfied in general.
    To satisfy this condition, a nonlinear feedback is required to cancel the complex nonlinearity of the plant system just as feedback linearization.
    On the other hand, the proposed controller requires Assumption~\ref{asm_positive_Lambda} which is more easily to be satisfied.
    This relaxation makes it possible to apply the proposed passivity-based sliding mode controller to various problem settings and particularly to trajectory tracking control as presented in the next section.
    The proposed controller also has additional free parameters that can be used to adjust control performance such as alleviating the chattering.
    Moreover, canceling the nonlinearity of the system requires precise information about the system, so the proposed controller is expected that it is more robust against the modeling error than the previous one.
    The proposed controller is the generalized version of the prototype one.
\end{remark}
 
\begin{remark}\label{rmk_contribution_SMC}
    In most cases of multiple input multiple output (MIMO) systems, many sliding mode controllers are designed so that each element of sliding variables $\sigma_i$'s are decoupled as in \eqref{BSMC_Dynamics_sigma}.
    On the other hand, Theorem~\ref{thm_stabilization} shows that sliding mode control can be achieved if 
    \begin{align*}
        \dot{\sigma} = -A\,\mathrm{sgn}\,\sigma
    \end{align*}
    holds with a positive definite matrix $A\in\mathbb{R}^{m\times m}$.
    This implies that sliding mode is realized in the MIMO case even if the matrix $A$ is not diagonal.
    This result is useful in the design of usual MIMO sliding-mode controllers.
\end{remark}
\begin{figure}
    \centering
    \subfigure[Potential function $U(\sigma)$]{%
        \includegraphics[clip, width=0.4\columnwidth]{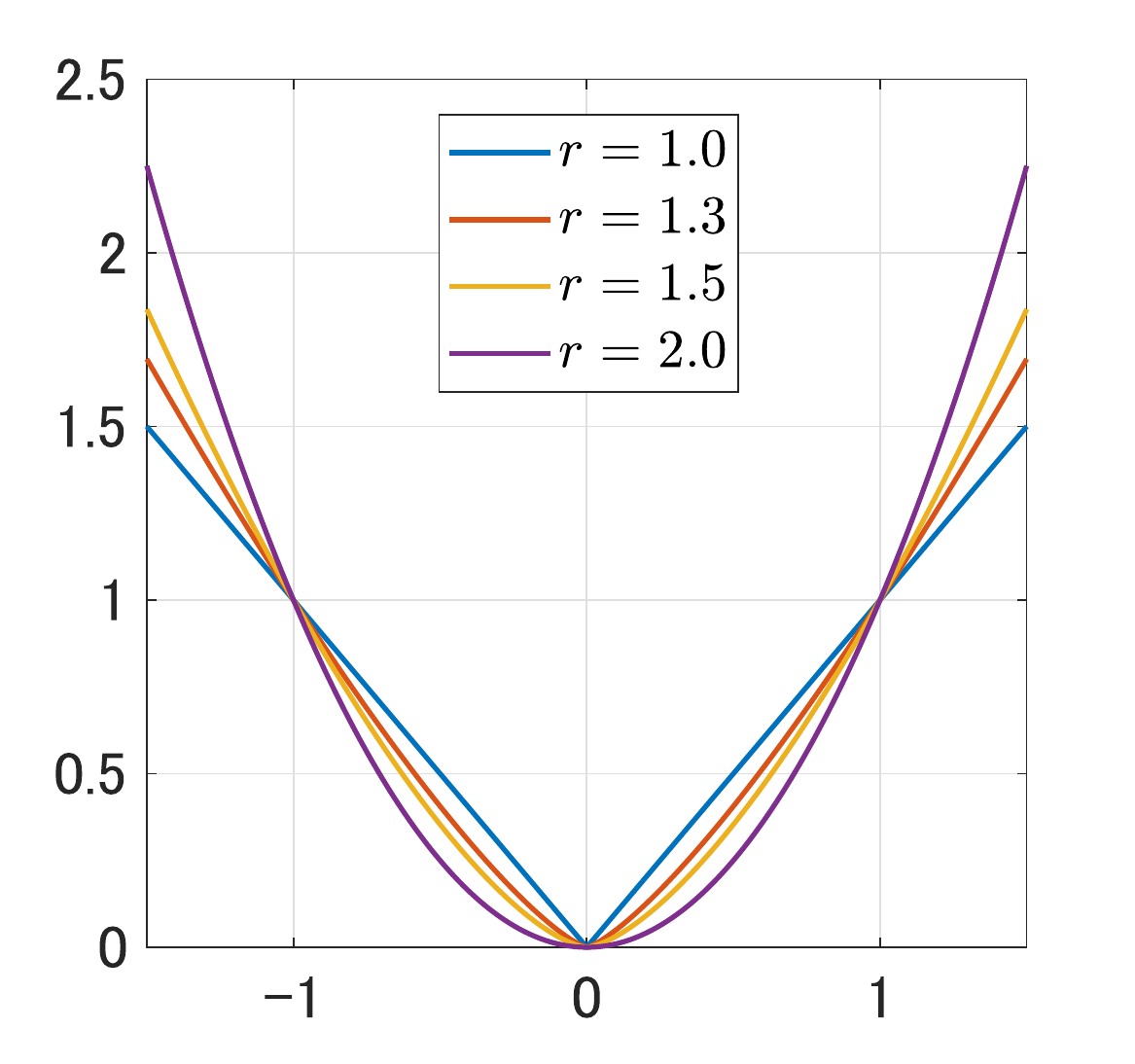}}%
    \subfigure[input $u$]{%
        \includegraphics[clip, width=0.4\columnwidth]{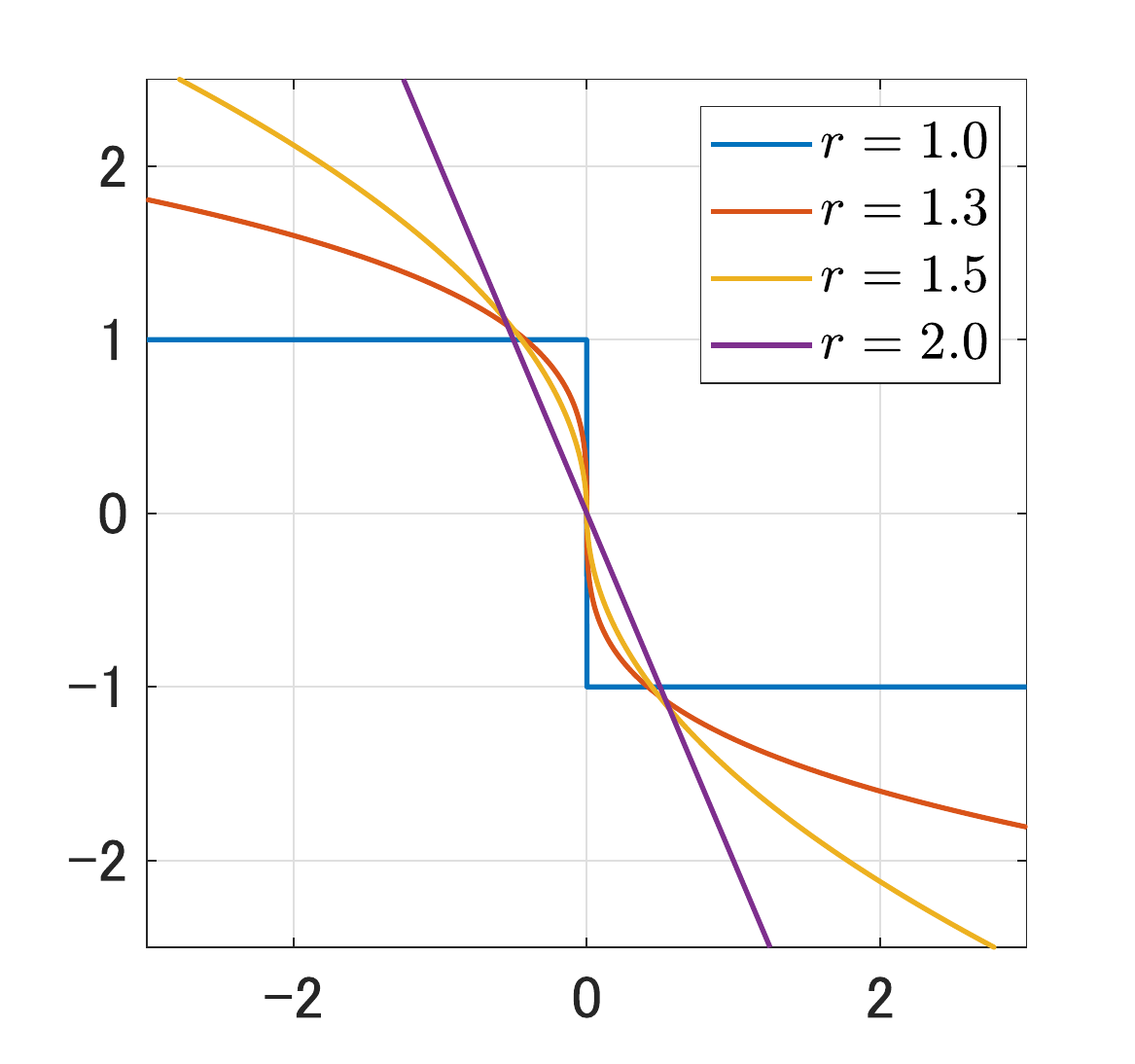}}%
    \caption{Potential function $U(\sigma) =\|\sigma\|^r$ and corresponding input $u = \nabla_\sigma U(\sigma)$.}
    \label{Variable_potential_and_input}
\end{figure}

The next section gives application of the proposed controllers to trajectory tracking control problems.

\section{Application to trajectory tracking control}
Trajectory tracking control for mechanical port-Hamiltonian systems has been studied, e.g., \cite{Fujimoto2003},\cite{Romero2015},\cite{Ferguson2019}.
In this section, we propose passivity-based sliding mode controllers for trajectory tracking control problems by combining these methods and the proposed PBSMC.

\subsection{Coordinate transformation to error system}
This subsection gives the closed-loop port-Hamiltonian system whose state variable is the tracking error between the state and its desired value.

For the port-Hamiltonian system \eqref{Square_PH_sys}, let us consider the desired trajectory $q^\mathrm{d} (t)$ of the configuration variable $q$ which is assumed to be bounded and twice differentiable.
Let us define the target momentum $\eta^\mathrm{d}(q,t)$ by
\begin{align*}
    \eta^\mathrm{d} (q,t) \equiv T(q)^{-1}\dot{q}^\mathrm{d} (t),
\end{align*} 
and define the error coordinate of the state variable by 
\begin{align}
    \begin{pmatrix} \tilde{q} \\ \tilde{\eta} \end{pmatrix} = 
    \begin{pmatrix} q -q^\mathrm{d}(t) \\ \eta - \eta^\mathrm{d}(q,t) \end{pmatrix}.
\end{align}
The following lemma shows the system \eqref{Square_PH_sys} is converted into the error port-Hamiltonian system.
\begin{lemma}
    Consider the system \eqref{Square_PH_sys} with the input
    \begin{align}
            u &= G(q)^{-1}\Bigl( D(q,\eta)\eta^\mathrm{d}(q,t) + \frac{\partial \eta^\mathrm{d}(q,t)}{\partial q}T(q)\eta + T(q)^{-1}\ddot{q}^\mathrm{d}(t) + v \Bigr)  \label{Feedback_Org2Err}
    \end{align}
    with $v \in \mathbb{R}^n$.
    Then the system is converted into the following closed-loop system with Hamiltonian function $\tilde{H}(\tilde{\eta}) = (1/2)\|\tilde{\eta}\|^2$
    \begin{align}
        \begin{pmatrix} \dot{\tilde{q}} \\ \dot{\tilde{\eta}} \end{pmatrix} &= 
        \left( \begin{array}{@{\,}c@{\>}c@{\,}}
            0_n & T(q) \\ -T(q)^\top &  -D(q,\eta)
        \end{array} \right)
        \begin{pmatrix} \nabla_{\tilde{q}} \tilde{H}(\tilde{\eta}) \\ \nabla_{\tilde{\eta}} \tilde{H}(\tilde{\eta}) \end{pmatrix} + 
        \begin{pmatrix} 0_n \\ I_n \end{pmatrix} v. \label{Error_PH_Sys}
    \end{align}
\end{lemma}
\begin{proof}
    The time derivative of $\tilde{q}$ is calculated as
    \begin{align}
        \dot{\tilde{q}} &= \dot{q} - \dot{q}^{\mathrm{d}}(t) \notag \\
                        &= T(q)\eta - T(q)\eta^{\mathrm{d}}(t) = T(q)\tilde{\eta} \label{dot_tilde_q}.
    \end{align}
    On the other hand, the time derivative of $\eta^{\mathrm{d}}(q,t)$ is derived as
    \begin{align}
        \dot{\eta}^{\mathrm{d}}(q,t) &= \frac{\partial \eta^\mathrm{d}(q,t)}{\partial q}\dot{q} + T(q)^{-1}\ddot{q}^{\mathrm{d}}(t) \notag \\
        &= \frac{\partial \eta^\mathrm{d}(q,t)}{\partial q}T(q)\eta + T(q)^{-1}\ddot{q}^{\mathrm{d}}(t) \label{dot_eta_d}.
    \end{align}
    Substituting \eqref{Feedback_Org2Err} into the system \eqref{Square_PH_sys} and using \eqref{dot_eta_d}, we obtain
    \begin{align}
        \dot{\eta} &= -D(q,\eta)\eta + G(q)u \notag \\
        &= -D(q,\eta)(\eta - \eta^\mathrm{d}(q,t)) + \dot{\eta}^{\mathrm{d}}(q,t) + v \notag \\
        \dot{\tilde{\eta}} &= -D(q,\eta)\tilde{\eta} + v. \label{dot_tilde_p}
    \end{align}
    Equations \eqref{dot_tilde_q} and \eqref{dot_tilde_p} coincide with \eqref{Error_PH_Sys}.
    Therefore, the input \eqref{Feedback_Org2Err} transforms the system \eqref{Square_PH_sys} into the closed-loop system \eqref{Error_PH_Sys}.
\end{proof}
This result is similar to the results in \cite{Fujimoto2003},\cite{Ferguson2019}.
The symbol $v$ denotes an additional input to be designed later.
Using this lemma, we can obtain an error coordinate system \eqref{Error_PH_Sys}, where the Hamiltonian function $\tilde{H}$ is independent of $\tilde{q}$.
Therefore, the proposed passivity-based sliding mode control technique can also be applied to the error system \eqref{Error_PH_Sys} in the same manner.
The next subsection gives the passivity-based sliding mode controller for trajectory tracking control.

\subsection{Controller design}
In this section, let us apply the proposed controller to trajectory tracking control problems.
Similarly to the previous section, we adopt the following assumption.
\begin{assumption}\label{asm_positive_Lambda_Tracking}
    The symmetric matrix $\tilde{\Lambda}(q, \tilde{q},\tilde{\eta})$ defined by
    \begin{align}
        \tilde{\Lambda}(q,\tilde{q},\tilde{\eta}) &\equiv 
        \frac{\partial \phi(\tilde{q}, \tilde{\eta})}{\partial \tilde{q}} T(q) \frac{\partial \phi(\tilde{q}, \tilde{\eta})}{\partial \tilde{\eta}}^\top  + \frac{\partial \phi(\tilde{q},\tilde{\eta})}{\partial \tilde{\eta}}T(q)^\top \frac{\partial \phi(\tilde{q},\tilde{\eta})}{\partial \tilde{q}}^\top \label{def:Lambda_tilde}
    \end{align}
    satisfies
    \begin{align}
        \tilde{\Lambda}(q,\tilde{q}, \tilde{\eta}) \succ \varepsilon I_m, \forall q, \tilde{q}, \tilde{\eta} \label{asm_posLambda_track}
    \end{align}
    with $\varepsilon > 0$
\end{assumption}

The following theorem gives a passivity-based sliding mode controller for trajectory tracking control. 

\begin{theorem}\label{thm_tracking}
Consider the error coordinate system \eqref{Error_PH_Sys} with any function $U: \mathbb{R}^m \to \mathbb{R}$ and a vector function $\phi: \mathbb{R}^m \times \mathbb{R}^m \to \mathbb{R}^m$ satisfying $\phi(0,0)=0$.
Suppose that the partial derivatives $\partial \phi(\tilde{q},\tilde{\eta})/\partial \tilde{q}$ and $\partial \phi(\tilde{q},\tilde{\eta})/\partial \tilde{\eta}$ of $\phi(\tilde{q}, \tilde{\eta})$ are nonsingular for any $\tilde{q}$ and $\tilde{\eta}$.
Then the feedback controller 
\begin{align}
    v = &- \tilde{\Lambda}(q,\tilde{q},\tilde{\eta}) \nabla_{\sigma}U(\sigma) + D(q,\eta)\tilde{\eta} - \frac{\partial \phi(\tilde{q}, \tilde{\eta})}{\partial \tilde{\eta}}^{-1}\frac{\partial \phi(\tilde{q}, \tilde{\eta})}{\partial \tilde{q}}T(q) \tilde{\eta} \label{Feedback_v}
\end{align}
converts the system \eqref{Error_PH_Sys} into the following closed-loop port-Hamiltonian system
\begin{align}
    &\begin{pmatrix} \dot{\tilde{q}} \\ \dot{\tilde{\eta}} \end{pmatrix} = 
    \begin{pmatrix}
        -T\frac{\partial \phi}{\partial \tilde{\eta}}^\top \frac{\partial \phi}{\partial \tilde{q}}^{-\top}& T \\ -T^\top & -\frac{\partial \phi}{\partial \tilde{\eta}}^{-1}\frac{\partial \phi}{\partial \tilde{q}}T    
    \end{pmatrix}
    \begin{pmatrix} \nabla_{\tilde{q}} \tilde{H}_{\mathrm{smc}} \\ \nabla_{\tilde{\eta}} \tilde{H}_{\mathrm{smc}} \end{pmatrix} \notag \\
    &\tilde{H}_{\mathrm{smc}}(\tilde{q},\tilde{\eta}) = \frac{1}{2}\|\tilde{\eta}\|^2 + U(\phi(\tilde{q}, \tilde{\eta})).
\label{SMC_Error_PH_sys}
\end{align}
Furthermore, if Assumptions~\ref{asm_U} and \ref{asm_positive_Lambda_Tracking} hold, then the sliding variable $\sigma = \phi(\tilde{q}, \tilde{\eta})$ is enforced to converge to zero after a finite transient and the origin of the closed-loop system \eqref{SMC_Error_PH_sys} is asymptotically stable with the Lyapunov function $\tilde{H}_{\mathrm{smc}}(\tilde{q},\tilde{\eta})$.
\end{theorem}

\begin{proof}
    The partial derivatives of $U(\phi(\tilde{q}), \tilde{\eta})$ with respect to $\tilde{q}$ and $\tilde{\eta}$ are calculated as
    \begin{align}
        \frac{\partial U(\phi(\tilde{q},\tilde{\eta}))}{\partial \tilde{q}} = \frac{\partial U(\sigma)}{\partial \sigma}
        \frac{\partial \phi(\tilde{q},\tilde{\eta})}{\partial \tilde{q}}, \label{rnd_U_rnd_tilde_q}\\
        \frac{\partial U(\phi(\tilde{q},\tilde{\eta}))}{\partial \tilde{\eta}} = \frac{\partial U(\sigma)}{\partial \sigma}
        \frac{\partial \phi(\tilde{q},\tilde{\eta})}{\partial \tilde{\eta}} \label{rnd_U_rnd_tilde_p}.
    \end{align}
    By substituting the input \eqref{Feedback_v} into the error system \eqref{Error_PH_Sys} and using \eqref{rnd_U_rnd_tilde_q} and \eqref{rnd_U_rnd_tilde_p}, it is proved that the error system \eqref{Error_PH_Sys} is transformed into the closed-loop system \eqref{SMC_Error_PH_sys}.

    The finite convergence property of sliding variable $\sigma$ and asymptotic stability of the closed-loop system \eqref{SMC_Error_PH_sys} are also proved in the same way as in Theorem~\ref{thm_stabilization}.
    This completes the proof.
\end{proof}

This theorem shows that passivity-based sliding mode control can be applied to trajectory tracking control.
As pointed out in Remark~\ref{rmk_unnecessarity_of_diagonal}, the controller in \cite{Fujimoto2021} requires the matrix $\Lambda(q, \tilde{q}, \tilde{\eta})$ to be diagonal.
However, since $T(q)$ that comes from the inertia matrix $M(q)$ depends on $q$ and the partial derivatives of the free parameter $\phi(\tilde{q}, \tilde{\eta})$ depend on $\tilde{q}$ and $\tilde{\eta}$, it is quite difficult to find a parameter $\phi(\tilde{q},\tilde{\eta})$ so that  $\tilde{\Lambda}(q, \tilde{q}, \tilde{\eta})$ is diagonal except when $T(q)$ is a constant matrix.
The authors have also proposed the method that the complex nonlinearity of the plant system is canceled with an additional nonlinear feedback so that the condition holds \cite{Sakata2021}.
But such canceling the nonlinearity of the system does not work well in the presence of the modeling error.
On the other hand, the proposed controller requires the relaxed condition \eqref{asm_posLambda_track} without canceling the nonlinearity, so it can be applied to trajectory tracking control and is expected to be more robust against modeling errors and external disturbances.

In the next section, the robustness of the proposed control system is analyzed.
\section{Robustness analysis}
Sliding mode control is known as a robust control method against modeling uncertainties and disturbances.
In this section, we investigate the robustness of the proposed control system.
Let us consider the system \eqref{Original_PH_sys} with disturbances 
\begin{align}
    &\begin{pmatrix} \dot{q} \\ \dot{p} \end{pmatrix} = 
    \left(\begin{array}{@{}c@{\ }c@{}}
        0_m & I_m \\ -I_m &  -D_0(q,p)\!
    \end{array}\right)\!
    \begin{pmatrix} \nabla_{\!q} H_0 \\ \nabla_{\!p} H_0 \end{pmatrix} + 
    \begin{pmatrix} 0_m \\ G_0(q)\! \end{pmatrix} \! u + 
    \begin{pmatrix} d_{\mathrm{um}}\! \\ d_{\mathrm{m}} \end{pmatrix}\!, \notag \\
    &H_0(q,p) = \frac{1}{2}p^\top M(q)^{-1} p, \label{Original_PH_sys_with_disturbance}
\end{align}
where $d_{\mathrm{um}}\in\mathbb{R}^m$ and $d_{\mathrm{m}}\in\mathbb{R}^m$ represent unmatched and matched disturbance respectively.
Since we consider mechanical systems, it is unlikely that disturbances will enter the kinematics, but here we assume that sensor noise and other factors are influencing the kinematics.

For the system \eqref{Original_PH_sys_with_disturbance}, let us apply the feedback input \eqref{feedback_genPBSMC} that has free parameters $\phi$ and $U$.
In this section, a vector function that is linear in momentum $\eta$ is selected as $\phi$, so the sliding variable $\sigma$ is given by
\begin{align}
    \sigma = \phi(q, \eta) = \psi(q) + \eta, 
\end{align}
where $\psi:\mathbb{R}^m \to \mathbb{R}^m$ is a diffeomorphism satisfying $\psi(0) =0$.
This choice satisfies the condition of $\phi(q,\eta)$ in Theorem~\ref{thm_stabilization}.
Note that the matrix $\Lambda$ defined by \eqref{DefLambda} becomes a function of $q$ as 
\begin{align*}
    \Lambda(q) = \frac{\partial \psi(q)}{\partial q}T(q) + T(q)^\top  \frac{\partial \psi(q)}{\partial q}^\top.
\end{align*}
Then the feedback input is calculated as
\begin{align}
    u = -G(q)^{-1} \!\left(\Lambda(q) \nabla_\sigma U(\sigma) - \!\left(D(q,\eta)+\frac{\partial \psi(q)}{\partial q}T(q)\!\right)\eta\right)\!,
\end{align}
and we obtain the following closed-loop port-Hamiltonian system with disturbances
\begin{align}
    &\begin{pmatrix} \dot{q} \\ \dot{\eta} \end{pmatrix} = 
    \begin{pmatrix}
        -T(q)\frac{\partial \psi(q)}{\partial q}^{-\top} & T(q) \\ 
        -T(q)^\top & -\frac{\partial \psi(q)}{\partial q}T(q)
    \end{pmatrix}
    \begin{pmatrix} \nabla_{q} H_{\mathrm{smc}} \\ \nabla_{\eta} H_{\mathrm{smc}} \end{pmatrix}  +
    \begin{pmatrix} d_{\mathrm{um}} \\ T(q)^\top d_{\mathrm{m}} \end{pmatrix}, \notag \\
    &H_{\mathrm{smc}}(q,\eta) = \frac{1}{2}\|\eta\|^2 + U(\psi(q) + \eta). \label{PBSMC_sys_with_disturbance}
\end{align}
For the system, the following theorem holds.
\begin{theorem}
    Consider the closed-loop system with disturbance \eqref{PBSMC_sys_with_disturbance}.
    The state of the control system \eqref{PBSMC_sys_with_disturbance} converges to the set defined by 
    \begin{align}
        B_1\equiv \left\{ x\in\mathbb{R}^{2n}\,\!\left|  \, \left\|\begin{pmatrix} \nabla_\eta U(\psi(q)+\eta) \\ \nabla_\eta U(\psi(q)+\eta) + \eta   \end{pmatrix}\right\|^2 \!\leq \gamma_1(q) \|d\|^2  \!\right.\right\}\!, 
    \end{align}
    where
    \begin{align}
        \gamma_1 (q) = \frac{4\left\{ \max \left[\lambda_{\max}\!\left(\frac{\partial \psi(q)}{\partial q}\frac{\partial \psi(q)}{\partial q}^\top\right)\!, \lambda_{\max}(T(q)\!^\top T(q))\right] \!\right\}}{\lambda_{\min}(\Lambda(q))^2}.
    \end{align}
    Moreover, if 
    \begin{align}
        \lambda_{\min}(\Lambda(q))\|\nabla_\eta U(\psi(q)+\eta)\| - \|\tfrac{\partial \psi(q)}{\partial q}d_{\mathrm{um}} + T(q)^{\!\top} d_{\mathrm{m}}\| > \delta \label{asm_nonzero_delU}
    \end{align}
    holds with a positive constant $\delta$, then the state of the control system \eqref{PBSMC_sys_with_disturbance} converges to the set 
    \begin{align}
        B_2 \equiv \left\{ x\in\mathbb{R}^{2n}\,|\,\{\|\eta\|^2 \leq \gamma_2(q) \|d_{\mathrm{um}}\|^2\} \wedge \{\psi(q) + \eta =0 \}\right\}
    \end{align}
    where 
    \begin{align}
        \gamma_2(q)= \frac{4\lambda_{\max}\left(\frac{\partial \psi(q)}{\partial q} \frac{\partial \psi(q)}{\partial q}^\top\right)}{\lambda_{\min}(\Lambda(q))^2}.
    \end{align}
\end{theorem}

\begin{proof}
    In this proof, the following symbols are used for simple notation
    \begin{align*}
        \overline{\lambda}(q) &\equiv \max\!\left(\lambda_{\max}(T(q)^\top T(q)), \lambda_{\max}\!\left(\frac{\partial \psi(q)}{\partial q}\frac{\partial \psi(q)}{\partial q}^\top\right) \!\right)\!,\\
        \overline{\lambda}_{\psi}(q) &\equiv \lambda_{\max}\left(\frac{\partial \psi(q)}{\partial q} \frac{\partial \psi(q)}{\partial q}^\top \right),\ \underline{\lambda}_{\Lambda}(q) \equiv \lambda_{\min}(\Lambda(q)).
    \end{align*}
    The time derivative of $H_{\mathrm{smc}}$ satisfies
    \begin{align*}
        \dot{H}_{\mathrm{smc}} 
        &= -\frac{1}{2}\begin{pmatrix} \nabla_q H_{\mathrm{smc}} \\ \nabla_\eta H_{\mathrm{smc}} \end{pmatrix}^{\!\!\top}\!\! 
        \begin{pmatrix}
            \frac{\partial \psi}{\partial q}^{-1} \Lambda \frac{\partial \psi}{\partial q}^{-\top} & 0_n \\ 0_n & \Lambda
        \end{pmatrix}\!\!
        \begin{pmatrix} \nabla_q H_{\mathrm{smc}} \\ \nabla_\eta H_{\mathrm{smc}} \end{pmatrix} + \begin{pmatrix} \nabla_q H_{\mathrm{smc}} \\ \nabla_\eta H_{\mathrm{smc}} \end{pmatrix}^{\top}
        \begin{pmatrix} d_{\mathrm{um}} \\ T^\top d_{\mathrm{m}} \end{pmatrix} \\
        &\leq -\frac{1}{2}\underline{\lambda}_{\Lambda}(\|\nabla_\eta U\|^2 + \|\nabla_\eta H_{\mathrm{smc}}\|^2 ) + \frac{c}{2}\left\| \begin{pmatrix} \frac{\partial \psi}{\partial q}^\top\nabla_\eta U\\ T\,\nabla_\eta H_{\mathrm{smc}} \end{pmatrix} \right\|^2 + \frac{1}{2c}\|d\|^2 \\
        &\leq -\frac{1}{2}(\underline{\lambda}_{\Lambda} - c\overline{\lambda})\left\| \begin{pmatrix} \nabla_\eta U\\ \nabla_\eta H_{\mathrm{smc}} \end{pmatrix} \right\|^2 + \frac{1}{2c}\|d\|^2 \\
        &\leq -\frac{\underline{\lambda}_{\Lambda} - c\overline{\lambda}}{2}
         \left(
            \left\| \begin{pmatrix} \nabla_\eta U\\ \nabla_\eta H_{\mathrm{smc}} \end{pmatrix} \right\|^2  - 
            \frac{1}{2c(\underline{\lambda}_{\Lambda} - c\overline{\lambda})}\|d\|^2
        \right)
    \end{align*}
    where $c\in \mathbb{R}$ is an arbitrary positive constant which is a result of the application of Young's inequality in the second line.
    The last inequality implies the Hamiltonian function $H_{\mathrm{smc}}$ decreases monotonically if 
    \begin{align*}
        \left\| \begin{pmatrix} \nabla_\eta U(q,\eta)\\ \nabla_\eta H_{\mathrm{smc}(q,\eta)} \end{pmatrix} \right\|^2 \geq 
            \frac{1}{c(\underline{\lambda}_{\Lambda}(q) - c\overline{\lambda}(q))}\|d\|^2 \label{Robust_radius_with_c}
    \end{align*}
    holds.
    By minimizing the right-hand side of \eqref{Robust_radius_with_c}, we can evaluate the convergence set tightly.
    Thus, we take $c = \underline{\lambda}_{\Lambda}(q) /(2\overline{\lambda}(q))$, and then obtain
    \begin{align*}
        \dot{H}_{\mathrm{smc}} &\leq -\frac{\underline{\lambda}_{\Lambda}(q)}{4}\left( \left\| \begin{pmatrix} \nabla_\eta U(q,\eta) \\ \nabla_\eta H_{\mathrm{smc}}(q,\eta) \end{pmatrix} \right\|^2 - \frac{4\overline{\lambda}(q)}{\underline{\lambda}_{\Lambda}(q)^2}\|d\|^2\right),
    \end{align*}
    which means the state variable converges to the set $B_1$.
    This completes the former part of the proof.

    Next, let us consider the case that the assumption \eqref{asm_nonzero_delU} holds.
    Let us apply the coordinate transformation \eqref{coordinate_trans_to_sigma} as in Theorem~\ref{thm_stabilization}.
    Then the resulting closed-loop system is described as
    \begin{align*}
        &\begin{pmatrix} \dot{\sigma} \\ \dot{\eta} \end{pmatrix} = 
        \begin{pmatrix} -\Lambda & 0 \\ -\Lambda & -\frac{\partial \phi}{\partial q}T \end{pmatrix}\!\!
        \begin{pmatrix} \nabla_{\sigma} H_{\mathrm{smc}} \\ \nabla_{\eta} H_{\mathrm{smc}} \end{pmatrix}  + 
        \begin{pmatrix} \frac{\partial \phi}{\partial q}d_{\mathrm{um}} + T^\top d_{\mathrm{m}} \\T^\top d_{\mathrm{m}}\end{pmatrix},\\
        &H_{\mathrm{smc}} = \frac{1}{2}\|\eta\|^2 + U(\sigma).
    \end{align*}
    To prove $\sigma$ converges to zero in a finite time, let us consider $U$ as a Lyapunov function candidate.
    With the assumption \eqref{asm_nonzero_delU}, the time derivative of $U$ satisfies
    \begin{align}
        \dot{U} 
        & = \nabla_{\sigma} U^\top ( -\Lambda \nabla_{\sigma} U + \tfrac{\partial \psi}{\partial q}d_{\mathrm{um}} + T^\top d_{\mathrm{m}} ) \notag \\ 
        & \leq -\lambda_{\min}(\Lambda)\|\nabla_{\sigma} U\|^2 + \|\nabla_{\sigma} U\| \|\tfrac{\partial \psi}{\partial q}d_{\mathrm{um}} + T^\top d_{\mathrm{m}}\| \notag \\
        & \leq - \| \nabla_{\sigma} U \| \delta   \\
        & \leq -\delta^2 / \lambda_{\min}(\Lambda) < 0 \label{dot_U_robust}.
    \end{align}
    Integrating \eqref{dot_U_robust} with respect to time, we can prove that the sliding variable $\sigma$ converges to zero and the sliding mode occurs in a finite time.
    In addition, $\dot{\sigma}=\sigma=0$ holds in the sliding mode, and then the time derivative of $H_{\mathrm{smc}}$ is calculated as 
    \begin{align*}
        \dot{H}_{\mathrm{smc}} &= 
        \nabla_{\eta} H_{\mathrm{smc}}^\top (-\Lambda \nabla_{\sigma} H_{\mathrm{smc}} -\tfrac{\partial \psi}{\partial q}T \nabla_{\eta} H_{\mathrm{smc}} + T^\top d_{\mathrm{m}} )\\
        &= \eta^\top (-\tfrac{\partial \psi}{\partial q}d_{\mathrm{um}} - T^\top d_{\mathrm{m}}-\tfrac{\partial \psi}{\partial q}T \eta  + T^\top d_{\mathrm{m}} ) \\
        &\leq -\frac{1}{2} \eta^\top\Lambda \eta + \frac{c}{2}\|\tfrac{\partial \psi}{\partial q}^\top\eta\|^2 + \frac{1}{2c}\|d_\mathrm{um}\| \\
        &\leq -\frac{1}{2}(\underline{\lambda}_{\Lambda} - c \overline{\lambda}_{\psi})\|\eta\|^2 + \frac{1}{2c}\|d_{\mathrm{um}}\|^2
    \end{align*}
    where $c$ is an arbitrary positive constant resulting from applying Young's inequality.
    Similarly to the former part of this proof, we take $c = \underline{\lambda}_{\Lambda}/(2\overline{\lambda}_{\psi})$ so that the ball where the state variable converges becomes minimum.
    Then we obtain
    \begin{align*}
        \dot{H}_{\mathrm{smc}} 
        &\leq -\frac{\underline{\lambda}_{\Lambda}(q)}{4}\left(\|\eta\|^2 -\frac{4\overline{\lambda}_{\psi}(q)}{\underline{\lambda}_{\Lambda}(q)^2}\|d_{\mathrm{um}}\|^2\right)
    \end{align*}
    on the sliding surface $\sigma = \psi(q) + \eta = 0$, which means the state converges to the set $B_2$.
    This completes the proof. 
\end{proof}

This theorem shows that the control system achieves robust stability in the sense that the state variable converges to a set $B_1$ and/or $B_2$.
For example, if a potential function $U(q,\eta) = (q+\eta)^2/2$ is selected, the resulting set $B_1$ is represented as 
\begin{align*}
    \left\|\begin{pmatrix} q + \eta \\ q + 2\eta \end{pmatrix}\right\|^2 \leq \gamma_1(q)\|d\|^2.
\end{align*}
Or, as a sufficient condition, it can be rewritten as 
\begin{align*}
    \left\|\begin{pmatrix} q \\ \eta \end{pmatrix}\right\|^2 \leq \frac{7+\sqrt{45}}{2}\gamma_1(q)\|d\|^2.
\end{align*}
This result would be useful to analyze the control system \eqref{PBSMC_sys_with_disturbance} with a smooth approximation of sliding mode controllers.
It is known that for first order sliding mode controllers there is a tradeoff between the smoothness and robustness of the controller.
Although it is difficult to evaluate a convergence set with conventional sliding mode controllers when a smooth approximation of input is employed, the proposed controller shows the set to which the state of the control system \eqref{PBSMC_sys_with_disturbance} converges by employing the Hamilton function $H_{\mathrm{smc}}$ as a Lyapunov function.

In addition to that, if the assumption \eqref{asm_nonzero_delU} holds and there is no unmatched disturbance $d_{\mathrm{um}}$, the state converges to zero.
For example, $\|\nabla_\eta U(\phi(q)+\eta)\|\geq k$ holds with $k>0$ when $U(\phi(q)+\eta) = k\|\phi(q)+\eta\|_1$ is selected.
Then it is possible to satisfy the condition \eqref{asm_nonzero_delU} by selecting a large value of $k$, which represents the property that sliding mode controllers become robust by increasing the input gain.

In the next section, we demonstrate the effectiveness of the proposed controller.

\section{Numerical example}
In this section, we show the effectiveness of the proposed controller through numerical simulations.
The plant system is a fully-actuated two degrees of freedom planar manipulator arm shown in Fig~\ref{arm}.
The control objective is to make the position of the end of the arm follow the target trajectory.
\begin{figure}[ht]
    \centering
    \includegraphics[width = 0.4\columnwidth]{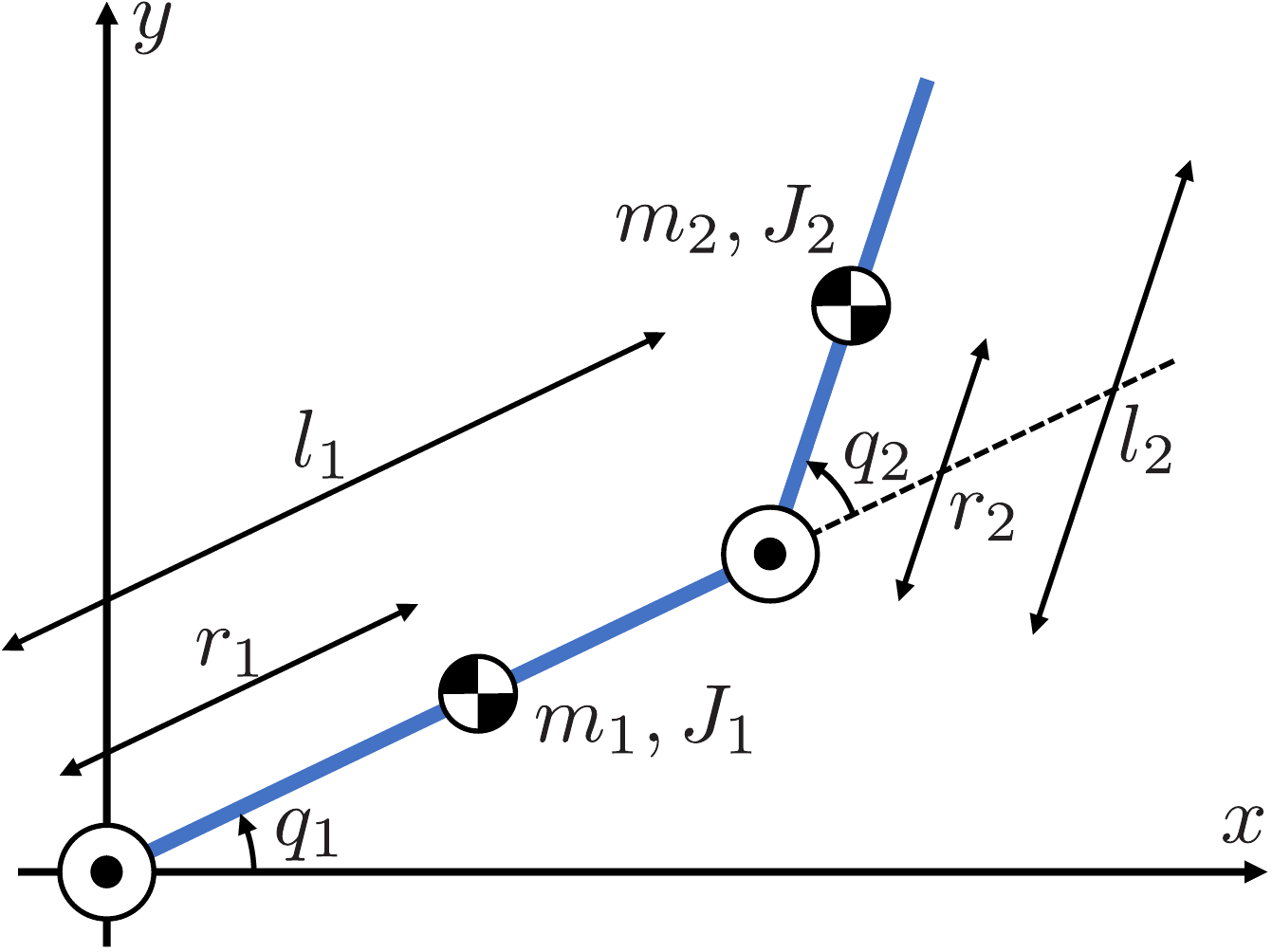}
    \caption{A two degrees of freedom manipulator arm}
    \label{arm}
\end{figure}
The equation of motion of the manipulator arm is described in the port-Hamiltonian representation \eqref{Original_PH_sys} as
\begin{align*}
    \begin{pmatrix} \dot{q} \\ \dot{p} \end{pmatrix} &= 
    \begin{pmatrix} 0_2 & I_2 \\ -I_2 & -D \end{pmatrix}
    \begin{pmatrix} \nabla_q H \\ \nabla_p H \end{pmatrix} +
    \begin{pmatrix} 0_2 \\ I_2 \end{pmatrix}u \\
    H(q,p) &= \frac{1}{2}p^\top M(q)^{-1} p.
\end{align*}
The state of this system consists of the angles of the links and the angular momenta, and we denote them as $q,\ p=M(q)\dot{q}\in \mathbb{R}^2$, respectively.
The inertia matrix $M(q)$ is given by
\begin{align*}
    M(q) &= \begin{pmatrix} M_1 + M_2 + 2M_3 \cos q_2 & M_2 + M_3\cos q_2 \\ M_2 + M_3\cos q_2  & M_2\end{pmatrix},\\
    M_1 &= m_1r_1^2 + m_2l_1^2 + J_1,\quad 
    M_2 = m_2r_2^2 + J_2, M_3 = m_2l_1r_2.
\end{align*}
Here, $m_i$ and $J_i$ denote the mass of the $i$-th link and the moment of inertia of the $i$-th link respectively, the symbol $l_i$ denotes the length of the $i$-th link, and $r_i$ denotes the length from the joint to the center of mass of the $i$-th link.
The damping matrix $D =\mathrm{diag}\,(\nu_1, \nu_2)$ consists of the friction coefficients $\nu_1, \nu_2$.
In this simulation, the physical parameters of the system are given in Table 1.
\begin{table}[b]
\caption{Physical parameters}
\centering
\begin{tabular}{ccc}
        & Link~1 & Link~2 \\ \hline
        Length of link & $l_1=1$ & $l_2=1$\\
        Mass of link & $m_1=1$ & $m_2=1$\\
        Center of mass of link & $r_1=1/2$ & $r_2=1/2$\\
        Moment of inertia of link & $J_1=1/12$ & $J_2=1/12$ \\
        Friction coefficient & $\nu_1=1/2$ & $\nu_2=1/2$
\end{tabular}
\end{table}
For this system, the matrix $T(q)$ satisfying \eqref{Definition_T} is given by
\begin{align*}
    &T(q) = \begin{pmatrix} \tfrac{\sqrt{M_2}}{\sqrt{M_1M_2-M_3^2\cos^2q_2}} & 0 \\ -\tfrac{M_2+M_3\cos q_2}{\sqrt{M_2}\sqrt{M_1M_2-M_3^2\cos^2q_2}} & \tfrac{1}{\sqrt{M_2}} \end{pmatrix},
\end{align*}
which is the result of Cholesky decomposition.

The controller is given by \eqref{Feedback_Org2Err} and \eqref{Feedback_v} where there are the free parameters $\phi(\tilde{q},\tilde{\eta})$ and $U(\sigma)$.
We select a function which depends linearly on $\tilde{q}$ and $\tilde{\eta}$ as $\phi$, and thus the sliding variable $\sigma$ is given by
\begin{align*}
    \sigma &= \phi(\tilde{q},\tilde{\eta}) = \begin{pmatrix} ~2~ & ~0~ \\ ~2~ & ~2~ \end{pmatrix}\begin{pmatrix} \tilde{q}_1 \\ \tilde{q}_2 \end{pmatrix} + \tilde{\eta}.
\end{align*}
In this case, the matrix $\tilde{\Lambda}$ defined by \eqref{def:Lambda_tilde} is calculated as
\begin{align*}
    \tilde{\Lambda}(q) = \frac{\sqrt{3}}{\sqrt{16-9\cos^2 q_2}}\begin{pmatrix} 4 & -3\cos q_2 \\ -3\cos q_2 & 12 \end{pmatrix} \succ 0,
\end{align*}
which satisfies Assumption~\ref{asm_positive_Lambda_Tracking}.
As mentioned in Section~3, general $p$-norm of the sliding variable $U(\sigma) = k\|\sigma\|^r_s$ satisfying \eqref{norm-potential} can be selected as a potential function satisfying Assumption~\ref{asm_U}. 
This class of functions has two parameters $r$ and $s$. 
The free parameter $r$ define how the input is continuous and the parameter $s$ changes the behavior in the reaching mode.
In this simulation, two potential functions are selected as
\begin{align}
    U(p+\phi(q)) &= 2\|\sigma\|_2^{1.3} = 2\|p+\phi(q)\|_2^{1.3}, \label{Ex_s-norm_powered_r}\\
    U(p+\phi(q)) &= 2\|\sigma\|_1 = 2\|p+\phi(q)\|_1 \label{Ex_1-norm_powered_1}.
\end{align}

The desired position of the end of the manipulator arm is selected by
\begin{align*}
    \begin{pmatrix} x^\mathrm{d}(t) \\  y^\mathrm{d}(t) \end{pmatrix} = \begin{pmatrix} 1 + 0.5 \cos t \\ 0.5 \sin t \end{pmatrix}.
\end{align*}
By calculating the inverse kinematics, the desired trajectory of the angle is derived as
\begin{align*}
    q^\mathrm{d}(x^\mathrm{d}(t), y^\mathrm{d}(t)) = 
    \begin{pmatrix}
        \tan^{-1}\!\left(\tfrac{y^\mathrm{d}(t)}{x^\mathrm{d}(t)}\right) - \cos^{-1}\!\left(\tfrac{l_1^2-l_2^2 + r^{\mathrm{d}}(t)^2}{2 l_1 r^{\mathrm{d}}(t)}\right)\\
        \cos^{-1}\!\left( \tfrac{r^{\mathrm{d}}(t)^2 -l_1^2-l_2^2}{2 l_1 l_2} \right)
    \end{pmatrix},
\end{align*}
with $r^{\mathrm{d}}(t)^2 \equiv x^{\mathrm{d}}(t)^2 + y^{\mathrm{d}}(t)^2$.
The initial condition of the state is given by 
\begin{align*}
    (q(0)^\top, p(0)^\top)^\top = (0,0,0,0)^\top.
\end{align*}

Figures \ref{fig:angle}-\ref{fig:Hamiltonian} show the result of the numerical simulations.
Figures \ref{fig:angle} and \ref{fig:sliding_variable} show the responses of angles $q$'s and sliding variables $\sigma$'s.
In Figure \ref{fig:angle}, the solid lines denote the responses of the angles and the dashed-dotted lines denote the desired angles. 
This result shows the angles track the desired trajectory with both potential functions.
In Figure \ref{fig:sliding_variable}, sliding variables converge to zero in a finite time ($t\approx 1.0$), and we can see that the proposed controller works as a sliding mode controller.
In the case $U(\sigma)= 2\|\sigma\|_2^{1.3}$, the sliding variables converge simultaneously, whereas the sliding variables converge independently, in the case $U(\sigma)= 2\|\sigma\|_1$.
The behavior of the sliding variables in the reaching mode can be changed by tuning the parameter $s$.
Figure \ref{fig:input} shows the responses of the inputs $u$.
By comparing Figure~\ref{fig:input}(a) with Figure~\ref{fig:input}(b), there is chattering in the input with the parameter $r=1$ whereas it is alleviated with the parameter $r=1.3$.
As these results show, the free parameters $r$ and $s$ are used to adjust the control performance.
Figure \ref{fig:Hamiltonian} shows the responses of the Hamiltonian functions $H_{\mathrm{smc}}$, and they decrease monotonically as proved in Theorem~\ref{thm_stabilization}.

These results show that the proposed passivity-based controller achieves sliding mode control for mechanical port-Hamiltonian systems.
It guarantees Lyapunov stability by employing the Hamiltonian function as a Lyapunov function.
Moreover, Lyapunov stability of the closed-loop system does not destroy even if the discontinuous feedback input is approximated with a continuous function.

\begin{figure}
    \centering
    \subfigure[$U(\sigma)=2\|\sigma\|_2^{1.3}$]{%
        \includegraphics[width=0.4\columnwidth]{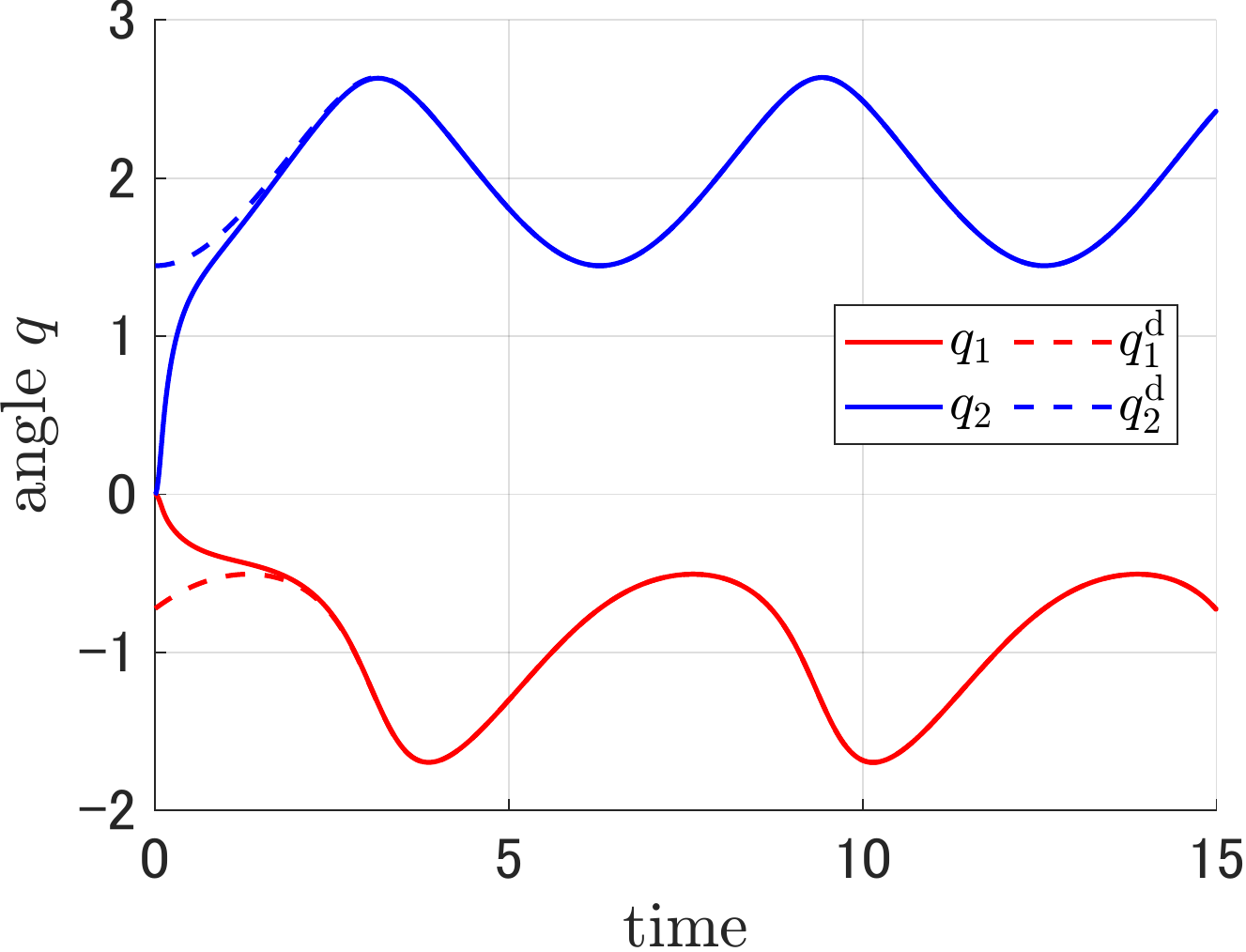}}%
    \subfigure[$U(\sigma) = 2\|\sigma\|_1$]{%
        \includegraphics[width=0.4\columnwidth]{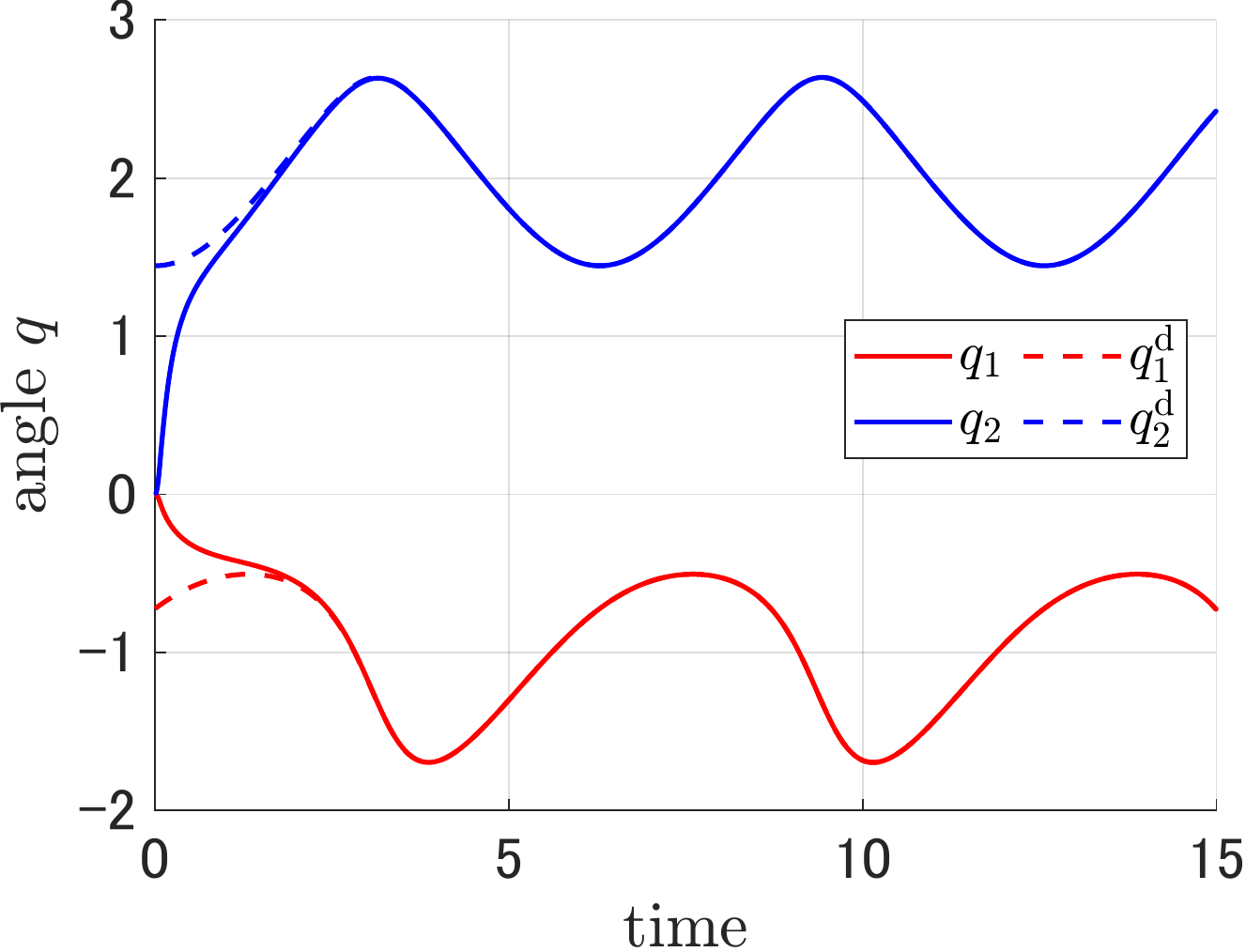}}%
    \caption{The responses of the angles}
    \label{fig:angle}
\end{figure}

\begin{figure}
    \centering
    \subfigure[$U(\sigma)=2\|\sigma\|_2^{1.3}$]{%
        \includegraphics[width=0.4\columnwidth]{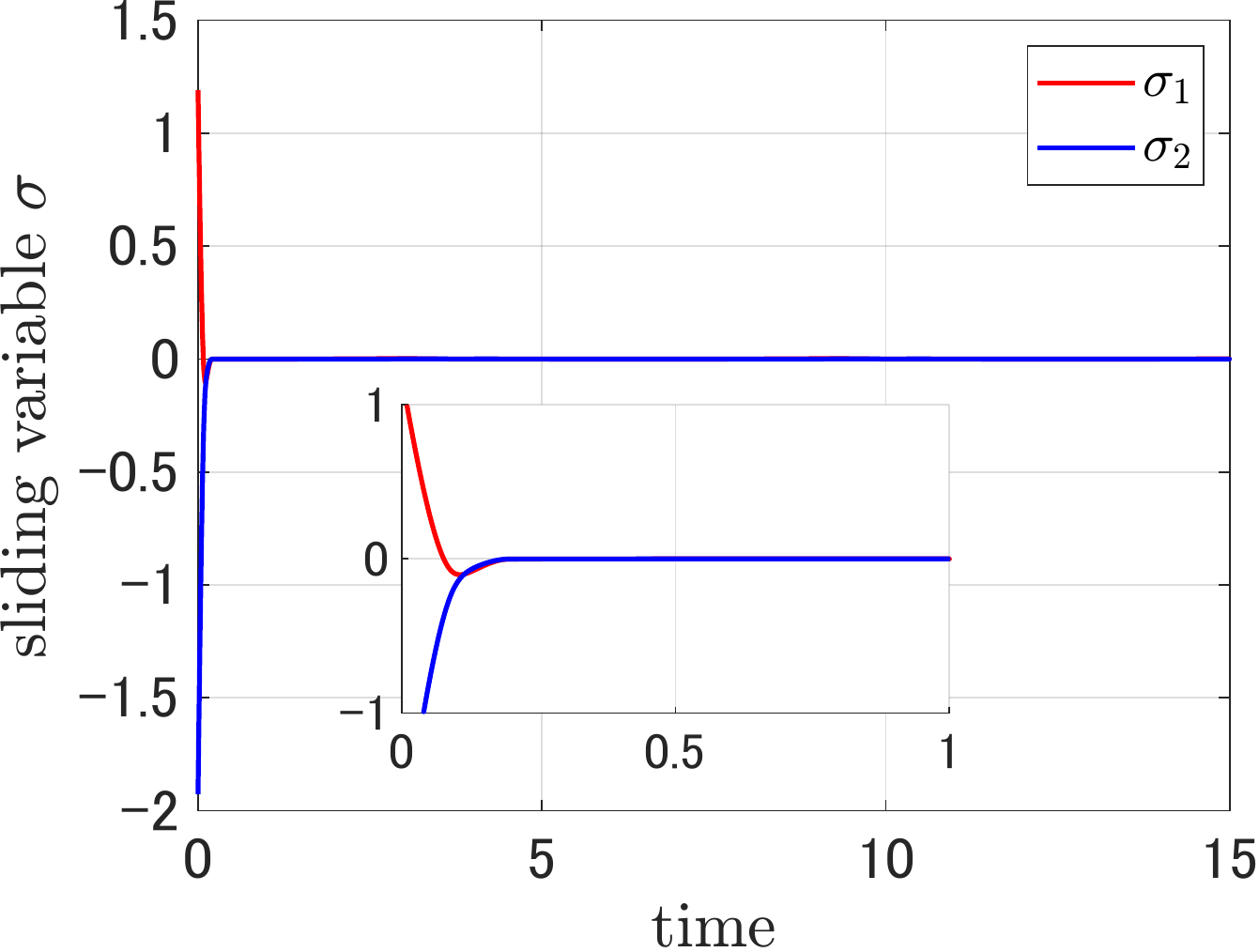}}%
    \subfigure[$U(\sigma) = 2\|\sigma\|_1$]{%
        \includegraphics[width=0.4\columnwidth]{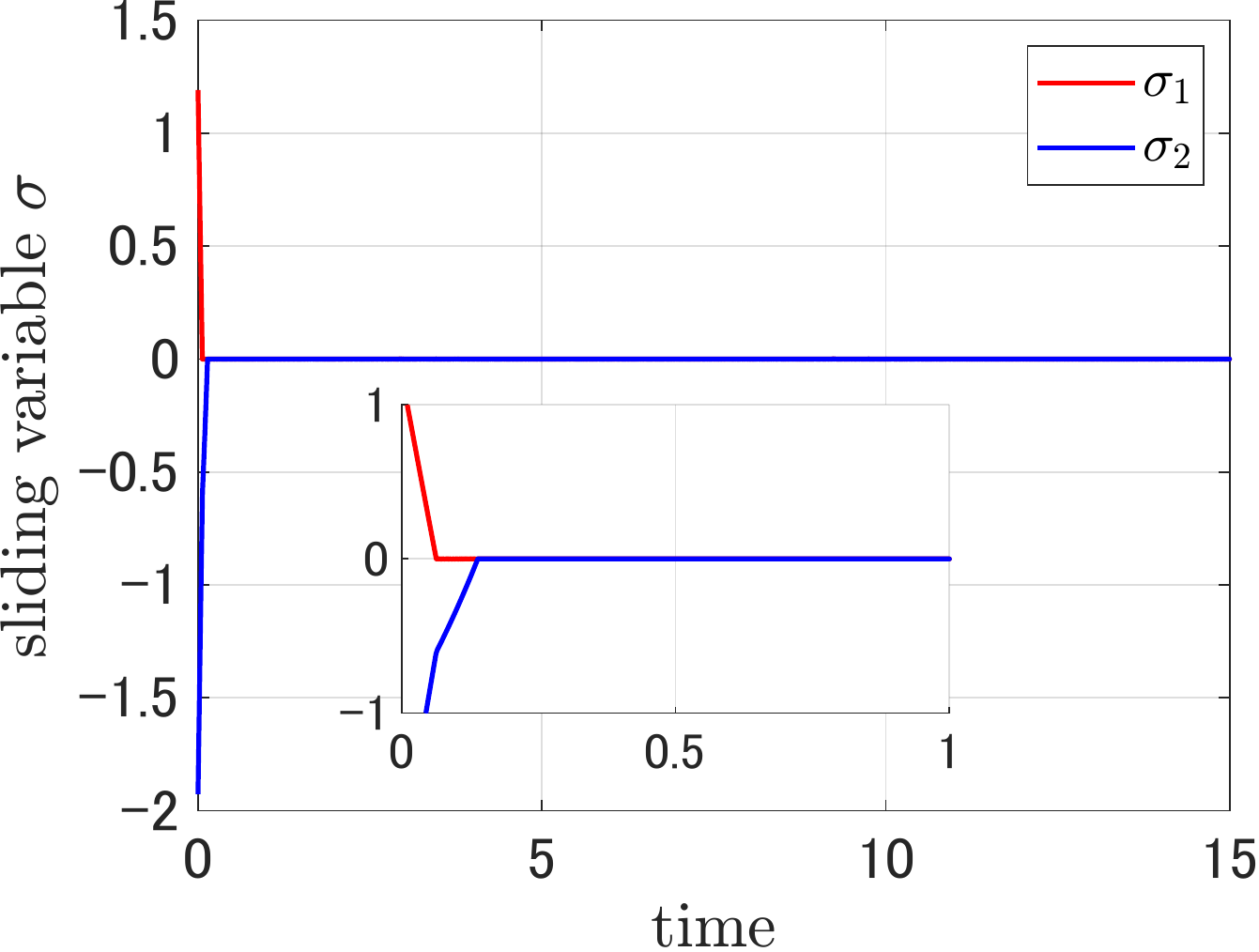}}%
    \caption{The responses of the sliding variables}
    \label{fig:sliding_variable}
\end{figure}

\begin{figure}
    \centering
    \subfigure[$U(\sigma)=2\|\sigma\|_2^{1.3}$]{%
        \includegraphics[width=0.4\columnwidth]{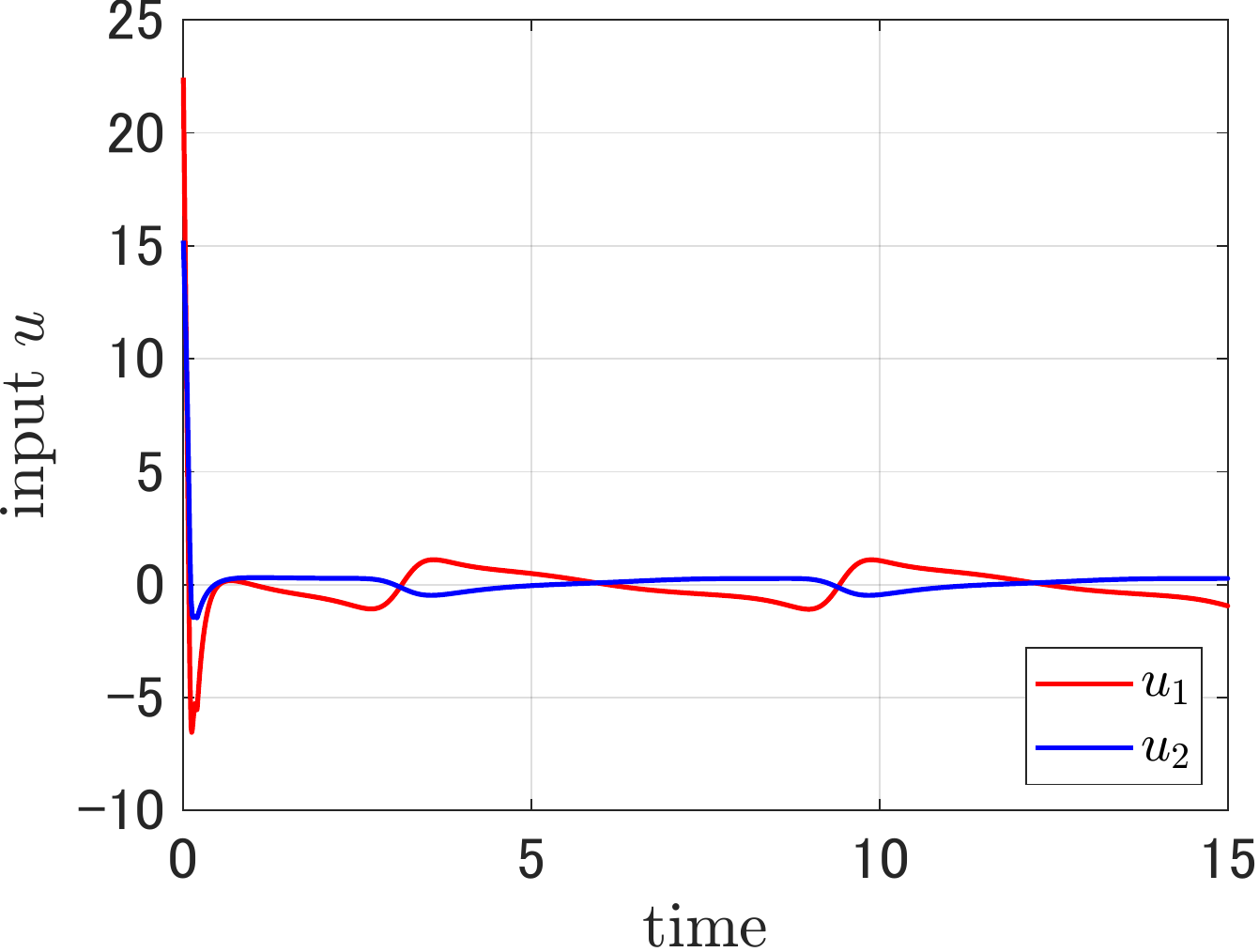}}%
    \subfigure[$U(\sigma) = 2\|\sigma\|_1$]{%
        \includegraphics[width=0.4\columnwidth]{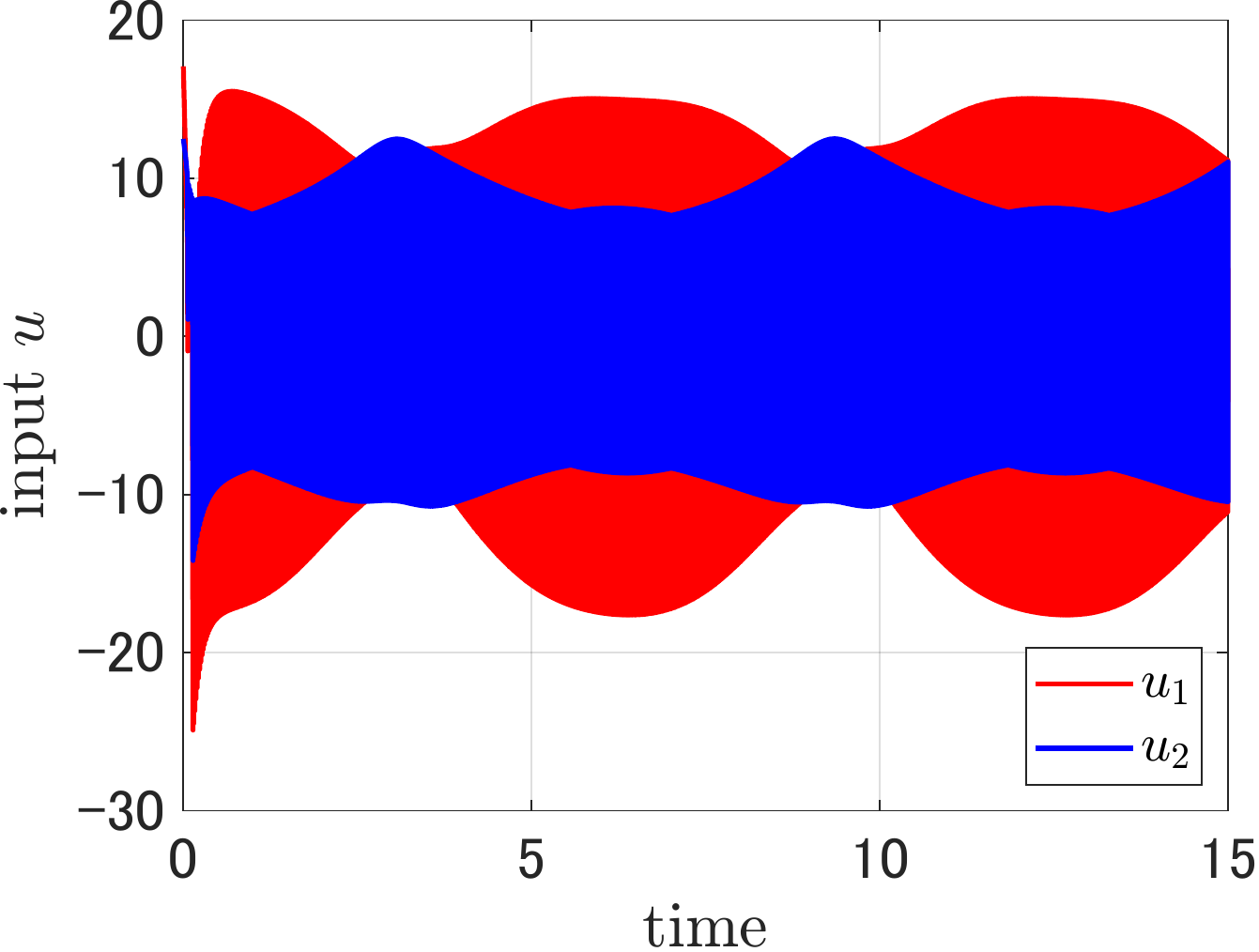}}%
    \caption{The responses of the inputs}
    \label{fig:input}
\end{figure}

\begin{figure}
    \centering
    \subfigure[$U(\sigma)=2\|\sigma\|_2^{1.3}$]{%
        \includegraphics[width=0.4\columnwidth]{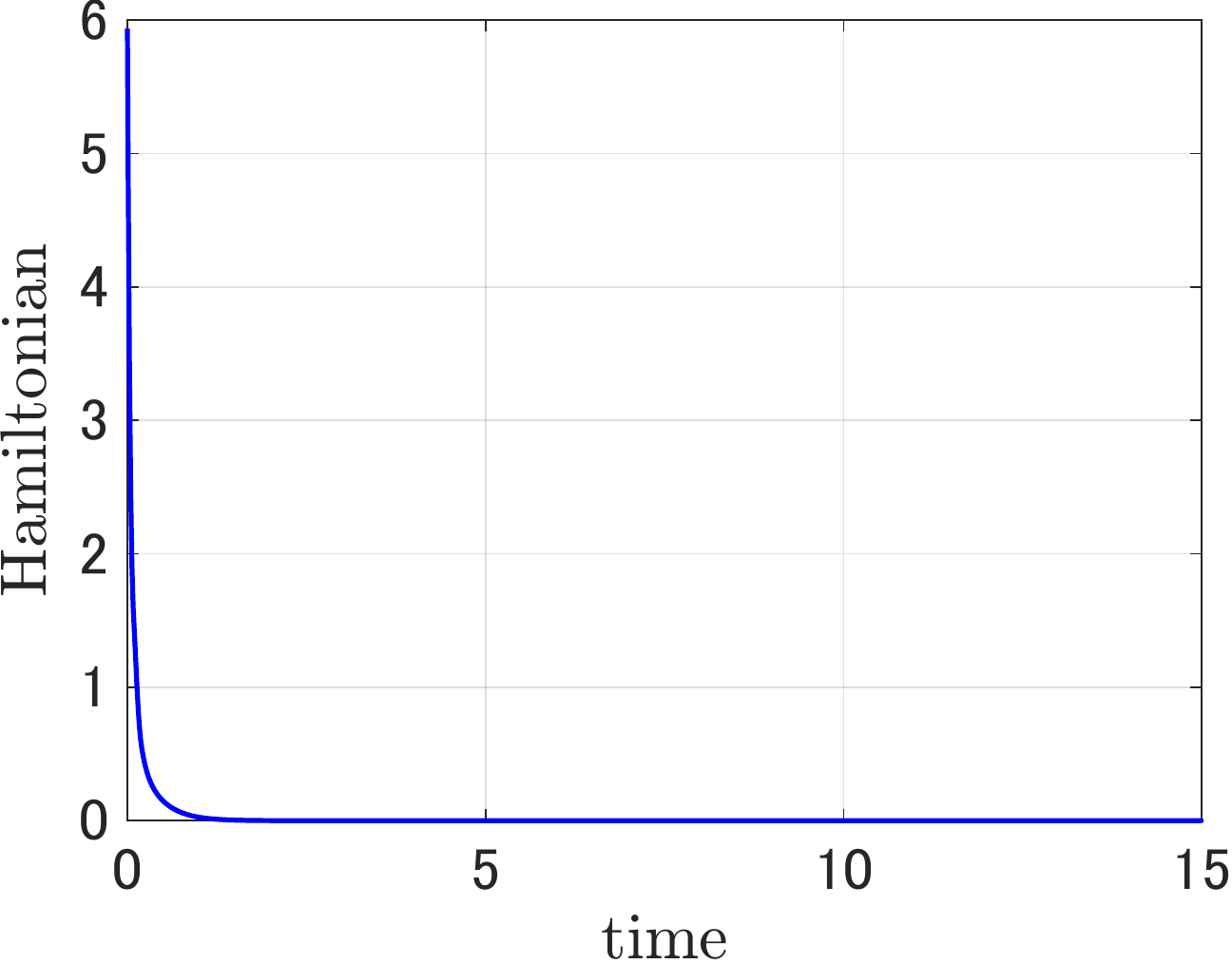}}%
    \subfigure[$U(\sigma) = 2\|\sigma\|_1$]{%
        \includegraphics[width=0.4\columnwidth]{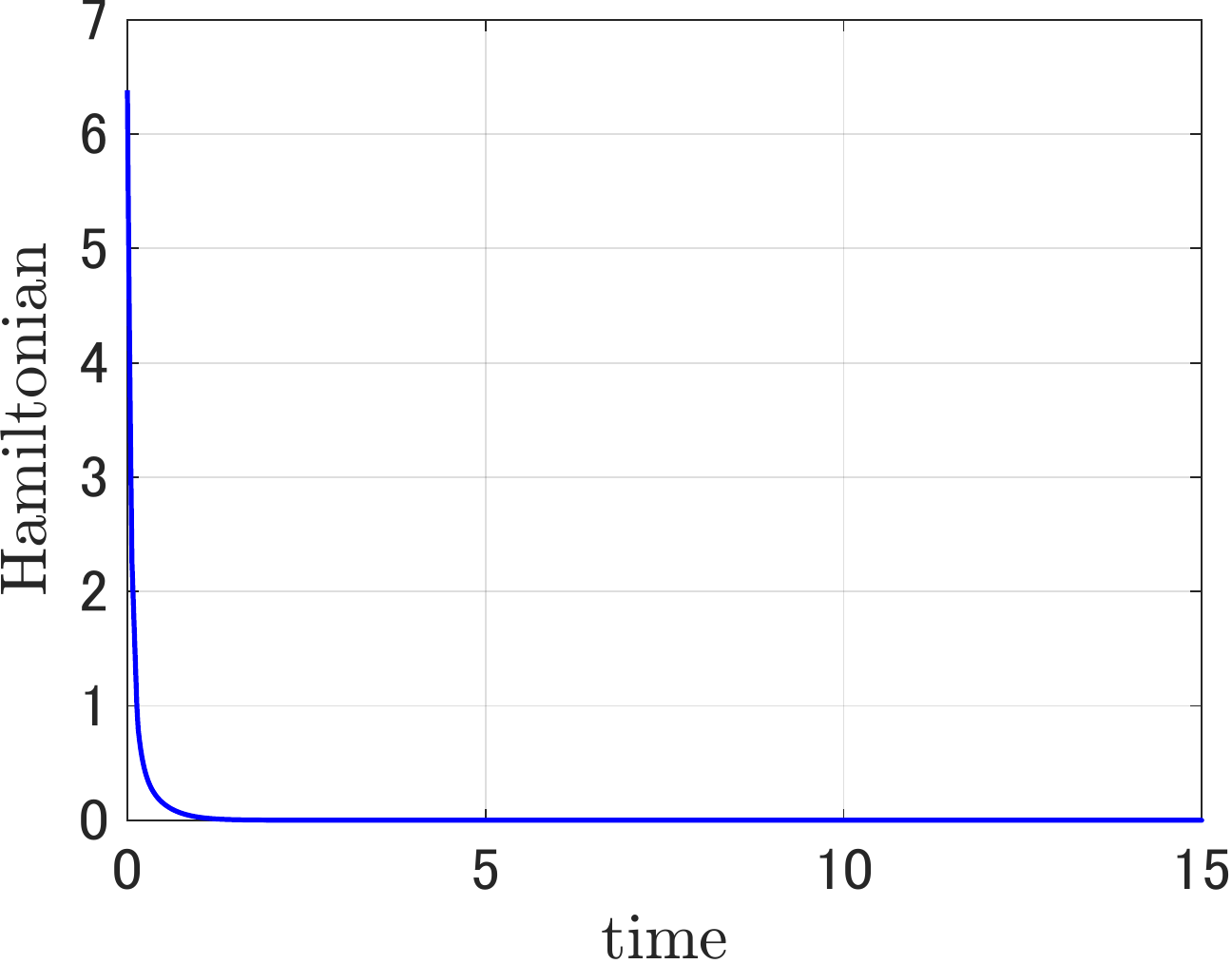}}%
    \caption{The responses of the Hamiltonian functions}
    \label{fig:Hamiltonian}
\end{figure}

\section{Conclusion}
This work proposes a new passivity-based controller that achieves sliding mode control with Lyapunov stability.
To construct the controller, we generalize KPES so that the freedom in selecting an artificial potential function increases.
The proposed controller represents stabilizing controllers consisting of both standard passivity-based controllers and sliding mode ones.
This makes it possible to adjust the trade-off between robustness against external disturbances and undesired chattering vibration while guaranteeing Lyapunov stability.
In addition, the controller can also be applied to trajectory tracking control.
Moreover, we analyze the robustness of the proposed control system against matched and unmatched disturbances.
The numerical examples demonstrate the effectiveness of the proposed method.





\end{document}